\documentclass[11pt, a4paper]{article} 

\usepackage[left=2.5cm,right=2.5cm,top=2cm,bottom=3cm]{geometry}

\usepackage[english]{babel}

\usepackage{bbm} 
\usepackage{amsmath}
\usepackage{amsfonts}
\usepackage{amssymb}
\usepackage{amsthm}
\usepackage{dsfont}
\usepackage{tikz}
\usetikzlibrary{calc,arrows.meta}
\usepackage{nicematrix}
\usepackage{pgfplots} 
\pgfplotsset{compat=1.16}
\usepgfplotslibrary{patchplots}

\usepackage{multirow}
\usepackage{adjustbox}
\usepackage{booktabs}
\usepackage{float}

\usepackage{caption}
\usepackage{subcaption}
\usepackage{threeparttable}

\usepackage{rotating}

\usepackage{natbib}
\usepackage{etoolbox}

\usepackage{appendix}

\usepackage{hyperref}

\usepackage{enumerate}

\newcommand{\E}{\mathrm{E}}

\newcommand{\p}{\mathbb{P}}

\newtheorem{theorem}{Theorem}[section]
\newtheorem{proposition}[theorem]{Proposition}

\newtheorem{lemma}[theorem]{Lemma}
\theoremstyle{definition}
\newtheorem{definition}[theorem]{Definition}
\newtheorem{assumption}{Assumption}
\newtheorem{example}[theorem]{Example}
\theoremstyle{remark}
\newtheorem{rem}[theorem]{Remark}

\usepackage[colorinlistoftodos,textsize=tiny]{todonotes}
\newcommand{\Comments}{1}
\newcommand{\mynote}[2]{\ifnum\Comments=1\textcolor{#1}{#2}\fi}
\newcommand{\mytodo}[2]{\ifnum\Comments=1%
	\todo[linecolor=#1!80!black,backgroundcolor=#1,bordercolor=#1!80!black]{#2}\fi}

\ifnum\Comments=1               
\fi

\ifnum\Comments=1               
\fi

\ifnum\Comments=1               
\fi

\begin{document}
	
	\title{Proper Correlation Coefficients for Nominal Random Variables\thanks{I thank Timo Dimitriadis, Daniel Gutknecht, Marc-Oliver Pohle, Christian H.\ Weiß, Johanna Nešlehová, Adrian Vetta and seminar participants at Goethe University, Frankfurt (2024, 2025), McGill University, Montreal (2025) and at the 12th HKMetrics workshop (2025) as well as conference participants at Statistische Woche, Wiesbaden (2025) for helpful comments.}}

	\author{Jan-Lukas Wermuth\thanks{Goethe University Frankfurt, RuW Building, Theodor-W.-Adorno-Platz 4, 60323 Frankfurt, Germany, e-mail: \href{mailto: wermuth@econ.uni-frankfurt.de}{wermuth@econ.uni-frankfurt.de}}}
	
	\maketitle

\begin{abstract}
This paper develops an intuitive concept of perfect dependence between two variables of which at least one has a nominal scale. Perfect dependence is attainable for all marginal distributions. It furthermore proposes a set of dependence measures that are 1 if and only if this perfect dependence is satisfied. The advantages of these dependence measures relative to classical dependence measures like contingency coefficients, Goodman-Kruskal's lambda and tau and the so-called uncertainty coefficient are twofold. Firstly, they are defined if one of the variables exhibits continuities. Secondly, they satisfy the property of attainability. That is, they can take all values in the interval $[0,1]$ irrespective of the marginals involved. Both properties are not shared by classical dependence measures which need two discrete marginal distributions and can in some situations yield values close to 0 even though the dependence is strong or even perfect.  

Additionally, the paper provides a consistent estimator for one of the new dependence measures together with its asymptotic distribution under independence as well as in the general case. This allows to construct confidence intervals and an independence test with good finite sample properties, as a subsequent simulation study shows. Finally, two applications on the dependence between the variables country and income, and country and religion, respectively, illustrate the use of the new measure.
\end{abstract}

	\newline
	
	
	\section{Introduction}
	\label{Introduction}
    Most random variables featuring in statistical applications are real-valued or exhibit at least an ordinal scale. Excluding stochastic processes, which can be interpreted as function-valued random variables, the only remaining type of random variables that appear as data in reality are those with a nominal scale. That is, they map to a set $\Omega'$ with finite cardinality consisting of elements that do not allow for any meaningful ordering (see Subsection \ref{subsec:setup} for details). While those variables are less interesting from a mathematical point of view as less measures of central tendency or dispersion and no cdf exists, they are still of immense practical relevance. 

    \begin{table}[tb]
	\caption{Source: \cite{WRD2025} (2020 data)} 
	\label{tab:countryrel_contingency_table}
    \begin{subtable}{.6\linewidth}
		\centering
        \subcaption{$3\times 3$ contingency table with absolute frequencies.}
		\label{tab:countryrel_contingency_tablea}
		\begin{tabular}{ccccc}
			\addlinespace
			\toprule
			& Christians &  Jews & Muslims &\\
			\midrule 
			Germany&56,071,000&127,000&5,351,000&61,549,000\\     
			Poland&36,782,000&3,100&39,200&36,824,300\\  
			Czechia&3,684,000&3,700&13,400&3,701,100\\
			&96,537,000&133,800&5,403,600&102,074,400\\
			\bottomrule
		\end{tabular}
	\end{subtable}
	\begin{subtable}{.39\linewidth}
		\centering
        \subcaption{Four dependence measures}
		\label{tab:countryrel_contingency_tableb}
		\begin{tabular}{cc}
			\addlinespace
			\toprule
			measure & value  \\
			\midrule 
			Cramér's V & 0.13 \\
			Pearson's C &  0.19\\
			Goodman-Kruskal's $\tau$ & 0.03  \\
			Uncertainty Coefficient&0.05\\
			\bottomrule
		\end{tabular}
	\end{subtable}
\end{table}
    In particular, the question of how to characterize statistical dependence between random variables when at least one variable is measured on a nominal scale remains largely unexplored in the literature. Most existing work on dependence assumes continuous marginals, since the copula representation is particularly well-behaved in this case \citep{sklar1959, nelsenIntroduction2006, genestEverything2007, durante2016principles}. Notable exceptions include \cite{geenens2020copula}, \cite{neslehova2007rank}, \cite{perroneGeometry2019}, \cite{genestEstimation2013}, and \cite{genest2007primer}. However, these studies focus on random variables whose values exhibit a reasonable ordering.
    
    For example, consider Table \ref{tab:countryrel_contingency_tablea} in which the absolute frequencies of Christians, Jews and Muslims living in Germany, Poland and Czechia are displayed. It is evident that the marginal frequencies do not allow for a concentration of either religion in one country, which would probably be the plainest type of (perfect) dependence possible. Instead, we have three primarily Christian countries which have Jewish and Muslim minorities of differing size. In fact, since Poland and Czechia have in the past enacted much stricter migration policies than Germany, the Muslims in this example are almost fully concentrated in Germany. The same holds true for Jews whose (re)settlement in Germany has been considered a great success after the Shoah. Taking the marginal frequencies as given, it is hard to imagine any bivariate distribution with a stronger dependence than the one displayed (besides removing Jews and Muslims completely out of Poland and Czechia). Note that fixing the marginals should always be the starting point for any dependence assessment since the two marginals and the dependence are three orthogonal objects, i.e.\ objects not carrying any information about each other \citep{geenens2023dependence}. Still, a selection of dependence measures in Table \ref{tab:countryrel_contingency_tableb} indicates little to no dependence. This raises the question to which extent those dependence measures sensibly measure the dependence inherent in this table.
    
    In mathematical terms, the nominal case creates several problems which explains why the stated dependence measures behave as they do. As no bivariate cdf exists, Fréchet-Hoeffding bounds \citep{frechet1951tableaux, hoeffding1940} and a (partial) dependence ordering \citep{tchen1980inequalities, yanagimoto1969partial} are not available. Consequently, dependence measures for the nominal case such as contingency coefficients, Goodman-Kruskal's lambda and tau and the uncertainty coefficient cannot adhere to such concepts which causes serious problems in the interpretation of these measures. For example, it is easily possible to construct marginal distributions which constrain the range of these coefficients from $[0,1]$ to $[0,c)$ with $0\le c\le1$ (see Section \ref{sec:improper_measures}). This is especially true for the introductory example and explains why the reported values are so small. However, since the marginal distributions have nothing to do with the dependence (see again \cite{geenens2023dependence}), any dependence measure showing such behavior is surely suboptimal.

    In addition to that, traditional measures critically rely on the assumption that both variables whose dependence shall be assessed are discrete. Thus, nominal-continuous comparisons such as country and income (on an individual level) or nominal-mixed comparisons such as country and precipitation (on a municipality level) are impossible. 
    
    The ambition of this paper therefore firstly is to offer a proposal for a definition of perfect dependence in the nominal case that is attainable irrespective of the marginals involved in Section \ref{sec:perfect_dependence}. Section \ref{sec:desirable_properties} collects a set of desirable properties whose fulfillment qualifies a dependence measure as \textit{proper}, one of those being that the measure attains the upper bound 1 if and only if this concept of perfect dependence is satisfied. To the best of my knowledge, existing sets of desirable properties are almost exclusively tailored towards real-valued random variables \citep{Renyi1959, Schweizer1981, scarsini1984measures, Mari2001, Embrechts2002, Balakrishnan2009, fissler2023generalised}. Appendix B in \cite{weiss2008measuring} marks a noteworthy exception.
    
    Section \ref{sec:improper_measures} shows the improperness of all existing dependence measures by exploiting their inability to deal with continuities in one marginal distribution and their non-attainability, i.e.\ their characteristic that the upper bound 1 can only be attained if the marginal distributions exhibit a specific structure. 

    Section \ref{sec:proper_measures} proposes a set of measures that cure the shortcomings of the improper measures in the nominal-nominal case and comply with the properties needed for properness. Additionally, it establishes an analogy between one of those new measures and the maximal correlation by \citet{gebelein1941maximal} and \citet{hirschfeld1935maximal}, an undirected dependence measure for real-valued random variables mapping to $[0,1]$. Moreover, it features a proposal for a dependence measure in the nominal-continuous case.
    
    Section \ref{sec:statistical_inference} defines a consistent estimator for the measure(s) based on Goodman-Kruskal's $\gamma$ \citep{Goodman1954} and computes its asymptotic distribution, both with and without the assumption of independence between the involved random variables. These results can in turn be used to construct confidence intervals and an independence test. Both methods of statistical inference will subsequently (in Section \ref{sec:sims}) be examined with respect to their finite sample performance in a simulation study that features a selected number of DGPs. 

    The paper closes with two applications in Section \ref{sec:case_study}, one concerning the dependence between country and income and the other concerning the dependence between country and religion. Section \ref{sec:conclusion} concludes. The Appendix extends Section \ref{sec:improper_measures}, falsifies an alternative concept of perfect dependence and explains which algorithms lead to a fast computation of the estimator. Additionally, it contains all the proofs as well as additional figures.

    \sloppy
    Moreover, I provide an \texttt{R} \citep{RCoreTeam} package \texttt{NCor} with an implementation of the estimator (by using the algorithms introduced in Appendix \ref{sec-App:Computation}), its confidence intervals and the independence test under \texttt{https://github.com/jan-lukas-wermuth/NCor}. The replication material for all the results in the paper is available under \texttt{https://github.com/jan-lukas-wermuth/replication\_NCor}.

	\section{Dependence Concepts for Nominal Random Variables} \label{sec:perfect_dependence}
	\subsection{Setup} 
    \label{subsec:setup}
	Consider two random variables $X$ and $Y$ defined on the same probability space $(\Omega, \mathcal{F}, \mathbb{P})$. On the one hand, at least one of the variables has a nominal scale, i.e.\ it maps to a set $\Omega'$ $(\Omega'')$ with $3\leq|\Omega'|<\infty$ $(3\leq|\Omega''|<\infty)$ consisting of arbitrary elements without a natural ordering.\footnote{It can, for example, be an alphabet, the set of all religions or the set of all nationalities.} We assume that these sets contain only elements with positive probability. The reason why the binary case is excluded is that you can always interpret such variables on an ordinal scale, i.e.\ as an event happening or not, and we have already treated this case in \citet{pohle2024measuring}. On the other hand, at most one of the variables exhibits at least an ordinal scale and therefore maps to $\mathbb{R}$ (or, at least, can be coded by elements of $\mathbb{R}$). Consequently, the marginal distribution(s) of the nominal variable(s) is (are) discrete, and the marginal distribution of the potential other random variable can be either discrete, continuous, or mixed. We are interested in measuring the strength of dependence between these two random variables. Any attempt to measure the direction of dependence is doomed due to the presence of at least one nominal random variable, which makes such an interpretation impossible.\footnote{\citet{weiss2008measuring} mention a possibility to define negative dependence also in the nominal case. However, a prerequisite for this definition is that the ranges of $X$ and $Y$ are identical. This requirement is considered too strict for our purposes.}
	
	\subsection{Perfect Dependence}\label{subsec:perfect_dependence}
	I now introduce a concept of perfect dependence for nominal random variables. Any such concept must be invariant with respect to the marginal distributions, i.e.\ it must be possible to attain perfect dependence irrespective of the shape and values that the involved marginals take \citep{pohle2024measuring, pohle2024}. The view that this paper holds is that dependence within real-valued random vectors is not fundamentally different from the dependence within the vectors considered here. Therefore, we can start our analysis with an arbitrary $\mathbb{R}^2$-valued random vector $(X,Y)$ with at least one discrete marginal distribution. We assume that this vector exhibits perfect positive or negative dependence. This is equivalent to its CDF $F_{X,Y}$ being identical to the Fréchet-Hoeffding upper bound $(F_{X,Y}(x,y):=\min\{F_X(x),F_Y(y)\})$ or the Fréchet-Hoeffding lower bound $(F_{X,Y}(x,y):=\max\{0,F_X(x)+F_Y(y)-1\})$, respectively. The replacement of the values of the discrete component(s) by elements of $\Omega'$ $(\Omega'')$ shall now only eliminate the directional interpretation. The attribute of perfectness however is sustained such that $(X,Y)$ is perfectly dependent, without any notion of positive or negative dependence.  
	
	\begin{definition}[Perfect Dependence]\label{def:perfect_dependence}Consider two cases:
		\begin{enumerate}
			\item Only $X$ (w.l.o.g.) has a nominal interpretation and we denote $|\Omega'|=:a\in \mathbb{N}_{\ge 3}$. We sort the elements of $\Omega'$ in an arbitrary but fixed way and write $\Omega'=\{\omega'_1, ..., \omega'_a\}$. Then, let $S_a$ be the set of all permutations of $\{1, ..., a\}$. With each $s_a\in S_a$, we associate $X^{s_a}$, the random variable $X$ with the numbering $s_a$ $(\omega'_1:=s_{a}(1), ..., \omega'_a:=s_{a}(a))$. $X$ and $Y$ are perfectly dependent if there exists an $s_a\in S_a$ such that $F_{X^{s_a},Y}$ is equivalent to either the Fréchet-Hoeffding lower or the Fréchet-Hoeffding upper bound.
			\item $X$ and $Y$ have a nominal interpretation and we denote $|\Omega'|=:a\in \mathbb{N}_{\ge 3}, |\Omega''|=:b\in \mathbb{N}_{\ge 3}$. We sort the elements of $\Omega'$ and $\Omega''$ in an arbitrary but fixed way and write $\Omega'=\{\omega'_1, ..., \omega'_a\}$ and $\Omega''=\{\omega''_1, ..., \omega''_b\}$. Then, let $S_a$ be the set of all permutations of $\{1, ..., a\}$ and $S_b$ be the set of all permutations of $\{1, ..., b\}$. With each $s_a\in S_a$, we associate $X^{s_a}$, the random variable $X$ with the numbering $s_a$ $(\omega'_1:=s_{a}(1), ..., \omega'_a:=s_{a}(a))$. Also, with each $s_b\in S_b$, we associate $Y^{s_b}$, the random variable $Y$ with the numbering $s_b$ $(\omega''_1:=s_{b}(1), ..., \omega''_b:=s_{b}(b))$. $X$ and $Y$ are perfectly dependent if there exist $s_a\in S_a$ and $s_b\in S_b$ such that $F_{X^{s_a},Y^{s_b}}$ is equivalent to either the Fréchet-Hoeffding lower or the Fréchet-Hoeffding upper bound.
		\end{enumerate}
	\end{definition}
	
	The following lemma shows that the above definition is unnecessarily complicated. In principle, we could drop one of the bounds from our definition.
	
	\begin{lemma}\label{lem:perfectdependence}
		In case 1 of Definition \ref{def:perfect_dependence}, the following two statements are equivalent:
		\begin{enumerate}[(i)]
			\item There exists a numbering $s_a\in S_a$ such that $F_{X^{s_a},Y}$ is equivalent to the Fréchet-Hoeffding lower bound.
			\item There exists a numbering $s_a\in S_a$ such that $F_{X^{s_a},Y}$ is equivalent to the Fréchet-Hoeffding upper bound.
		\end{enumerate}
		In case 2 of Definition \ref{def:perfect_dependence}, the following two statements are equivalent:
		\begin{enumerate}[(i)]
			\item There exist numberings $s_a\in S_a$ and $s_b\in S_b$ such that $F_{X^{s_a},Y^{s_b}}$ is equivalent to the Fréchet-Hoeffding lower bound.
			\item There exist numberings $s_a\in S_a$ and $s_b\in S_b$ such that $F_{X^{s_a},Y^{s_b}}$ is equivalent to the Fréchet-Hoeffding upper bound.
		\end{enumerate}
	\end{lemma}
	
	\section{Desirable Properties of Dependence Measures}
	\label{sec:desirable_properties}
	We want to summarize the dependence between $X$ and $Y$ in a single number, a dependence measure.
	\begin{definition}[Dependence Measure]
		A dependence measure for the non-deterministic random variables $X$ and $Y$ is a mapping $\delta:D\to \mathbb{R}\ (\mathcal{L}_{X,Y}\mapsto \delta(\mathcal{L}_{X,Y}))$, where we write $\delta(X, Y):=\delta(\mathcal{L}_{X,Y})$ and define $D:=\mathcal{M}_1(\Omega'\times \Omega'')$ or $D:=\mathcal{M}_1(\Omega'\times \mathbb{R})$. $\mathcal{M}_1(\cdot)$ denotes the set of all probability measures on the respective product set, equipped with either $\mathcal{P}(\Omega')\otimes \mathcal{P}(\Omega'')$ or $\mathcal{P}(\Omega')\otimes \mathcal{B}(\mathbb{R})$. 
	\end{definition}
	As already mentioned before, a fundamental difference to the case of two real-valued random variables is the impossibility to interpret directions of dependence. Thus, this paper aims to construct a dependence measure that indicates only the strength of dependence, but not the direction. In order to perform a substantial assessment of whether a measure deserves the attribute of \textit{properness}, we need a set of desirable properties.
	
	\begin{definition}[Proper Dependence Measure]\label{def:proper}
		We call a dependence measure $\delta(X, Y)$ for the non-deterministic random variables $X$ and $Y$ of which at least $X$ (w.l.o.g.) has a nominal scale \emph{proper}, if it fulfills the following characteristics:
		\begin{enumerate}[(A)]
			\item \emph{Existence}: $\delta(X,Y)$ exists for any pair of random variables $X$ and $Y$.
			\item \emph{Normalization}: $0\le \delta(X,Y)\le 1$.
			\item \emph{Independence}: $\delta(X, Y)=0$ if and only if $X$ and $Y$ are independent.
			\item \emph{Attainability}: $\delta(X, Y)=1$ if and only if $X$ and $Y$ are perfectly dependent.
			\item \emph{Absolute invariance to strictly monotonic transformations} (only applicable if $Y$ has at least an ordinal scale):
			\begin{align*}
				\delta(X,g(Y))=
					\delta(X,Y),\ \text{g: } \mathbb{R}\to \mathbb{R} \text{ strictly monotonic on the range of } Y
			\end{align*}
			\item \emph{Symmetry}: $\delta(X, Y)=\delta(Y, X)$.
		\end{enumerate}
	\end{definition}
	
	Relative to the real-valued case, the most remarkable difference is the absence of negative values in the co-domain of the dependence measure. This is natural because without the option to measure direction of dependence, there is no need to distinguish between positive and negative values. Additionally, a (partial) dependence ordering similar to \citet{yanagimoto1969partial} and \citet{tchen1980inequalities} is in my opinion difficult to construct and therefore left for future research. As a consequence, I also do not include any property formalizing compliance of a dependence measure with such a concept. Finally, the attainability axiom uses the novel definition of perfect dependence and is thus fundamentally different from the attainability definition in the real-valued case.
	
	\section{Improper Dependence Measures}
	\label{sec:improper_measures}
	This section together with Appendix \ref{sec:AppendixImproperMeasures} reviews the most popular dependence measures for the nominal case and thus extends Appendix B in \cite{weiss2008measuring}. More precisely, it considers several measures based on Pearson's mean square contingency coefficient. Appendix \ref{sec:AppendixImproperMeasures} contains measures based on proportional reduction of predictive error and an entropy-based measure. Measures of agreement such as Cohen's $\kappa$ \citep{cohen1960agreement} are left out of the discussion because their assumption that $X$ and $Y$ take exactly the same values is considered too strict for our purposes.
	\subsection{Measures Based on Pearson's Mean Square Contingency Coefficient}
\begin{table}[tb]
\centering
\begin{minipage}[t]{0.47\textwidth}
\centering
\caption{Google and Google Scholar search results (6th January 2026)}
\label{tab:measures_scholar}
\begin{tabular}{ccc}
    \toprule
    measure &  Google & Scholar    \\
    \midrule 
    Cram\'er's V & 320.000 & 29.800 \\
    Pearson's C & 95.300 & 26.200 \\
    Tschuprow's T & 1.940 & 245 \\
    Goodman-Kruskal's $\lambda$ & 1.110 & 44 \\
    Goodman-Kruskal's $\tau$ & 278 & 31 \\
    Uncertainty Coefficient &19.800 & 5.020\\
    \bottomrule
\end{tabular}
\end{minipage}
\hfill
\begin{minipage}[t]{0.47\textwidth}
\centering
\caption{An $a \times b$ contingency table with $p_{ij}:=\mathbb{P}(X=x_i, Y=y_j)$, $p_{i\cdot}:=\sum_{j=1}^bp_{ij}$ and $p_{\cdot j}:=\sum_{i=1}^ap_{ij}$.}
		\label{tab:contingency_table}
		\begin{tabular}{ccccc}
			\toprule
			& $y_1$ &  $\hdots$ & $y_b$ \\
			\midrule 
			$x_1$         &   $p_{11}$   & $\hdots$& $p_{1b}$  &$p_{1\cdot}$   \\
			$\vdots$           & $\vdots$  &   $\ddots$    &   $\vdots$   \\  
			$x_a$&   $p_{a1}$ & $\hdots$ & $p_{ab}$&$p_{a\cdot}$ \\
			&$p_{\cdot 1}$&$\hdots$&$p_{\cdot b}$&1\\
			\bottomrule
		\end{tabular}
\end{minipage}

\end{table}
	The most popular dependence measures for the nominal case arise as normalizations of Pearson's Mean Square Contingency $(MSC)$ coefficient (see Table \ref{tab:measures_scholar} and \citet{liebetrau1983measures} for an overview). This coefficient is the population analogue of the test statistic in Pearson's $\chi^2$ independence test. It is only defined if both marginal distributions are discrete, i.e.\ if $X$ takes the values $x_1, ..., x_a$ and $Y$ takes the values $y_1, ..., y_b$:
	\begin{align*}
    MSC(X,Y):=\sum_{i=1}^{a}\sum_{j=1}^{b}\frac{(\mathbb{P}(X=x_i, Y=y_j)-\mathbb{P}(X=x_i)\mathbb{P}(Y=y_j))^2}{\mathbb{P}(X=x_i)\mathbb{P}(Y=y_j)}
	\end{align*}
    
	Cramér's $V$ \citep{Cramer1945}, Tschuprow's $T$ \citep{Tschuprow1925}, Pearson's Contingency Coefficient $PC$ \citep{Pearson1904} and Sakoda's $S$ \citep{sakoda1977measures} are defined as 
	\begin{align*}
		V(X,Y):=&\sqrt{\frac{MSC(X,Y)}{\min\{a-1,b-1\}}}, \quad T(X,Y):=\sqrt{\frac{MSC(X,Y)}{\sqrt{(a-1)(b-1)}}}, \\
		PC(X,Y):=&\sqrt{\frac{MSC(X,Y)}{1+MSC(X,Y)}} \quad \text{and}\quad S(X,Y):=\sqrt{\frac{MSC(X,Y)\min\{a,b\}}{(1+MSC(X,Y))\min\{a-1,b-1\}}}.
	\end{align*}
	Even though their improperness is immediate due to the necessity to have two discrete marginal distributions, it is insightful to analyze their violation of the attainability property in more detail. The following lemma is well-known in the literature.
	
	\begin{lemma}[$MSC$ attainability]\label{lem:MSC_attainability}
		It holds that $0\le MSC(X,Y)\le \min\{a,b\}-1$. Additionally, the following two statements hold:
		\begin{itemize}
			\item We have $MSC(X,Y)=0$ if and only if $X$ and $Y$ are independent.
			\item We have $MSC(X,Y)=\min\{a,b\}-1$ if and only if $\mathbb{P}(X=x_i, Y=y_j)=\mathbb{P}(X=x_i)$ with some $j=1, ..., b$ holds for all $i=1, ..., a$ (in the case of $a\ge b$) or $\mathbb{P}(X=x_i, Y=y_j)=\mathbb{P}(Y=y_j)$ with some $i=1, ..., a$ holds for all $j=1, ..., b$ (in the case of $a\le b$).
		\end{itemize}
	\end{lemma}
	If we interpret the previous lemma with a contingency table with $a$ rows and $b$ columns (see Table \ref{tab:contingency_table}), the second statement is equivalent to each row (if $a\ge b$) or each column (if $a\le b$) containing only one non-zero entry. Equipped with this lemma, we can analyze the attainability properties of the other coefficients.
	\begin{lemma}[attainability]\label{lem:V_attainability}
		It holds that $0\le V(X,Y),T(X,Y),PC(X,Y),S(X,Y)\le 1$. Additionally, the following statements hold:
		\begin{itemize}
			\item We have $V(X,Y)=0$ if and only if $T(X,Y)=0$ if and only if $PC(X,Y)=0$ if and only if $S(X,Y)=0$ if and only if $X$ and $Y$ are independent.
			\item We have $V(X,Y)=1$ if and only if $S(X,Y)=1$ if and only if $MSC(X,Y)=\min\{a,b\}-1$.
			\item We have $T(X,Y)=1$ if and only if $MSC(X,Y)=\min\{a,b\}-1$ and $a=b$.
			\item We never have $PC(X,Y)=1$.
		\end{itemize}
	\end{lemma}
	Given the conditions under which the coefficients attain their maximum value of 1, it is evident that it is easily possible to construct marginal distributions which do not allow a fulfillment of these conditions (see e.g.\ the introductory example). Hence, they can be regarded as too strict as the following lemma shows.
	\begin{lemma}[dependence concepts]\label{lem:MSC_DC}
		If a discrete random vector $(X,Y)$ fulfills $MSC(X,Y)=\min\{a,b\}-1$, it also fulfills the concept of perfect dependence defined in Definition \ref{def:perfect_dependence}. The reverse statement is wrong in the following sense:
        \begin{enumerate}
            \item The concept of perfect dependence as defined in Definition \ref{def:perfect_dependence} may be satisfied with marginals that make $MSC(X,Y)=\min\{a,b\}-1$ impossible. 
            \item The concept of perfect dependence as defined in Definition \ref{def:perfect_dependence} may be satisfied with marginals that make $MSC(X,Y)=\min\{a,b\}-1$ possible. Still, it may hold that $MSC(X,Y)\ne \min\{a,b\}-1$.
        \end{enumerate}
	\end{lemma}
	
	\section{Proper Dependence Measures}
	\label{sec:proper_measures}
	The essential idea behind the construction of proper dependence measures is the concept of perfect dependence explained in Subsection \ref{subsec:perfect_dependence} together with property (D) from Definition \ref{def:proper}.
	
	\begin{definition}[Proper Measure]\label{def:wermuth_coefficient}
		Consider again the two cases from Definition \ref{def:perfect_dependence}:
		\begin{enumerate}
			\item Only $X$ (w.l.o.g.) has a nominal interpretation and we denote $|\Omega'|=:a\in \mathbb{N}_{\ge 3}$. We sort the elements of $\Omega'$ in an arbitrary but fixed way and write $\Omega'=\{\omega'_1, ..., \omega'_a\}$. Then, let $S_a$ be the set of all permutations of $\{1, ..., a\}$. With each $s_a\in S_a$, we associate $X^{s_a}$, the random variable $X$ with the numbering $s_a$ $(\omega'_1:=s_{a}(1), ..., \omega'_a:=s_{a}(a))$. The proposed measure is \begin{align*}
				\delta^{*}(X, Y)=\max_{s_a\in S_a}\delta(X^{s_a},Y)
			\end{align*}
			with some proper correlation coefficient $\delta$ \citep{pohle2024}.
			\item $X$ and $Y$ have a nominal interpretation and we denote $|\Omega'|=:a\in \mathbb{N}_{\ge 3}, |\Omega''|=:b\in \mathbb{N}_{\ge 3}$. We sort the elements of $\Omega'$ and $\Omega''$ in an arbitrary but fixed way and write $\Omega'=\{\omega'_1, ..., \omega'_a\}$ and $\Omega''=\{\omega''_1, ..., \omega''_b\}$. Then, let $S_a$ be the set of all permutations of $\{1, ..., a\}$ and $S_b$ be the set of all permutations of $\{1, ..., b\}$. With each $s_a\in S_a$, we associate $X^{s_a}$, the random variable $X$ with the numbering $s_a$ $(\omega'_1:=s_{a}(1), ..., \omega'_a:=s_{a}(a))$. Also, with each $s_b\in S_b$, we associate $Y^{s_b}$, the random variable $Y$ with the numbering $s_b$ $(\omega''_1:=s_{b}(1), ..., \omega''_b:=s_{b}(b))$. The proposed measure is \begin{align*}
				\delta^{*}(X, Y)=\max_{s_a\in S_a, s_b\in S_b}\delta(X^{s_a},Y^{s_b})\end{align*}
			with some proper correlation coefficient $\delta$ \citep{pohle2024}.
		\end{enumerate}
	\end{definition}
	
	Since Goodman-Kruskal's $\gamma$ is a proper dependence measure with straightforward inference procedures \citep{pohle2024inference}, the subsequent analysis will focus on this special case.
    	
	\begin{proposition}\label{prop:propercase1}
	Assume case 1 from Definition \ref{def:perfect_dependence}. Then, $\gamma^*(X,Y)$ fulfills all axioms from Definition \ref{def:proper} besides axiom (C). More precisely, $\gamma^*(X,Y)=0$ does not imply independence between $X$ and $Y$.
	\end{proposition}
    \begin{proposition}\label{prop:propercase2}
		Assume case 2 from Definition \ref{def:perfect_dependence}. Then, $\gamma^*(X,Y)$ is proper.
	\end{proposition}
	
	\subsection{Relationship to Undirected Dependence Measures for Real-Valued Random Variables}
	Even though $\gamma^*$ is a measure for nominal random variables, the idea behind its construction is less connected with existing measures for the nominal case than with existing measures for $\mathbb{R}^2$-valued random vectors that map to $[0,1]$. Basically, those measures sacrifice the directional interpretation for a stronger interpretation of the 0 (the measures are 0 if and only if $X$ and $Y$ are independent) and a wider concept of perfect dependence which allows those measures to detect U-shaped relationships, for example. One particularly similar coefficient has been introduced by \citet{hirschfeld1935maximal} and \citet{gebelein1941maximal} and is since called the maximal correlation. It is defined as
	\begin{align*}
		S_r(X,Y):=\sup_{f,g}r(f(X),g(Y)),
	\end{align*}
	where $r$ denotes the Pearson correlation and the supremum is taken over the set of all Borel-measurable functions $f,g:\mathbb{R}\to \mathbb{R}$. If we now consider 
	\begin{align*}
		S_{\gamma}(X,Y):=\sup_{f,g}\gamma(f(X),g(Y)),
	\end{align*}
	and restrict the set of functions over which we maximize to the set of all bijective functions, an interesting simplification emerges.
	\begin{lemma}\label{lem:maximalcorrelation}
    Assume $X$ and $Y$ are both discrete $\mathbb{R}$-valued random variables. It holds that 
		\begin{align*}
			S_{\gamma}(X,Y)=\max_{s_a\in S_a, s_b\in S_b}\gamma(X^{s_a}, Y^{s_b}),
		\end{align*}
		where $(X^{s_a}, Y^{s_b})$ is defined as in case 2 of Definition \ref{def:wermuth_coefficient}, albeit without the necessity to have a nominal interpretation for $X$ and $Y$.
	\end{lemma}
	Intuitively, the rank-based nature of $\gamma$ makes it sufficient to consider all possible orderings of the values that $X$ and $Y$ can take. In the case of two discrete distributions, the innovation of our proposal is now to consider nominal random variables together with functions $f:\Omega'\to \mathbb{R}$ (and $g: \Omega''\to \mathbb{R}$)\footnote{\citet{holzmann2024lancaster} already mention the ad hoc extension of the maximal correlation to $\mathcal{X}$- and $\mathcal{Y}$-valued random variables $X$ and $Y$, where $\mathcal{X}$ and $\mathcal{Y}$ are arbitrary measurable spaces.} and maximize $\gamma$, a correlation coefficient with superior properties relative to Pearson correlation \citep{pohle2024}.
	
	\section{Statistical Inference}
	\label{sec:statistical_inference}
	The population coefficient introduced in the previous section is only of practical use if it is possible to construct a consistent estimator for it. Additionally, knowledge about the asymptotic distribution could allow us to construct asymptotically valid confidence intervals as well as an independence test based on the coefficient. In the following, $\overset{d}{\to}$ denotes convergence in distribution and $\overset{p}{\to}$ denotes convergence in probability.
	\subsection{Estimator}
	Our estimator simply replaces the theoretical distribution $\mathcal{L}_{X,Y}$ implicitly used in Definition \ref{def:wermuth_coefficient} with the empirical distribution $$\widehat{\mathcal{L}}_{X,Y}:=\frac{1}{n}\sum_{i=1}^{n}\delta_{X_i,Y_i}.$$
    We make the following sampling assumption:
    \begin{assumption}\label{ass:iid}
        $\{X_i,Y_i\}_{i=1}^n$ is a bivariate i.i.d.\ sample.
    \end{assumption}
	\begin{definition}[Estimator]\label{def:estim}
		Consider again the two cases from Definition \ref{def:perfect_dependence}:
		\begin{enumerate}
			\item Only $X$ (w.l.o.g.) has a nominal interpretation. 
			The empirical marginal distribution $\widehat{\mathcal{L}}_X:=1/n\sum_{i=1}^{n}\delta_{X_i}$ has almost surely only a finite number of values with positive probability (relative frequency). We denote those values by $x_1, ..., x_k$, where it holds that $\{x_1, ..., x_k\}\subseteq (\omega_1', ..., \omega_a')$. Again, we redefine $x_1:=s_{k}(1), ..., x_k:=s_{k}(k)$ according to a permutation $s_k\in S_k$ and obtain a sample of $(X^{s_k},Y)$ with an artificially created random variable $X^{s_k}$. The maximum of the set of all $k!$ possible coefficients is our estimate:
			\begin{align*}
				\widehat{\gamma}^*_n(X,Y)=\max_{s_k \in S_k}\widehat{\gamma}_n(X^{s_k},Y).
			\end{align*}
			\item $X$ and $Y$ have a nominal interpretation. 
			The empirical marginal distributions $\widehat{\mathcal{L}}_X:=1/n\sum_{i=1}^{n}\delta_{X_i}$ and $\widehat{\mathcal{L}}_Y:=1/n\sum_{i=1}^{n}\delta_{Y_i}$ have almost surely only a finite number of values with positive probability (relative frequency). We denote those values by $x_1, ..., x_k$ and $y_1, ..., y_l$, respectively. It holds that $\{x_1, ..., x_k\}\subseteq (\omega_1', ..., \omega_a')$ and $\{y_1, ..., y_l\}\subseteq (\omega_1'', ..., \omega_b'')$. Again, we redefine $x_1:=s_{k}(1), ..., x_k:=s_{k}(k)$ and $y_1:=s_{l}(1), ..., y_l:=s_{l}(l)$ according to permutations $s_k\in S_k$ and $s_l\in S_l$ and obtain a sample of $(X^{s_k},Y^{s_l})$ with artificially created random variables $X^{s_k}$ and $Y^{s_l}$. The maximum of the set of all $k!\cdot l!$ possible coefficients\footnote{The estimation of this coefficient can be computationally challenging since the factorial increases very fast. For $k,l\le 7$, this is usually of no concern. However, the argument of computational feasibility is also the reason why we do not consider Bootstrap inference.\label{foot:computationalintensive}} is our estimate:
			\begin{align*}
				\widehat{\gamma}^*_n(X,Y)=\max_{s_k \in S_k, s_l \in S_l}\widehat{\gamma}_n(X^{s_k},Y^{s_l}).
			\end{align*}
		\end{enumerate}    
	\end{definition}
	
	\begin{proposition}\label{prop:gamma*_consistency}
		Under Assumption \ref{ass:iid}, $\widehat{\gamma}^*_n \overset{p}{\to} \gamma^*$.
	\end{proposition}
    Note that however consistent, the estimator will in general be biased in finite samples (see Figure \ref{fig:bias} for a simulation study which shows small mean biases for moderate sample sizes). Due to the small size of the biases and the strong limits that are placed upon bias correction techniques by \citet{hirano2012impossibility}, I refrain from proposing alternative estimators.\footnote{Also, the results by \citet{chernozhukov2013intersection} constitute no remedy for this problem since half median unbiasedness is of limited value for a point estimate.}
	
	\subsection{Asymptotic Distribution}
	Imagine the true (or one of potentially many true) population numbering(s) was known to the researcher. Then, an alternative estimator $\widehat{\gamma}(X^{s_k^{(*)}},Y)$ or $\widehat{\gamma}(X^{s_k^{(*)}},Y^{s_l^{(*)}})$ could be considered in which $s_k^{(*)}$ and $s_l^{(*)}$ are fixed. In this case, the standard asymptotic distribution of $\widehat{\gamma}$ could be used for statistical inference \citep{pohle2024inference}. Unfortunately, this is in reality never the case and we have no choice but revert to the estimator introduced in Definition \ref{def:estim}. Since the maximum function in this estimator always chooses one specific numbering and can (also asymptotically) oscillate between several of those, the asymptotic distribution is non-standard. In order to calculate it, we first need the asymptotic distribution of the vector of inputs at the left hand side of (\ref{eq:estim}). Since the length of this vector depends on the number of values that $X$ (and $Y$) can take, this paper can only provide a strategy to compute the asymptotic variance and no closed-form solution. 
	\begin{proposition}\label{prop:jointasymptotics}
		Consider the two cases of Definition \ref{def:estim}.\footnote{We define the estimator with respect to the number of values with positive probability mass in the theoretical distribution of $X$ (and $Y$) and not with respect to the number of distinct values that occur in the sample. However infeasible in practice, this is convenient for this proposition because there almost surely exists a sample size $N\in \mathbb{N}$ for which $k=a$ and $l=b$ holds. Thus, until the sample size $N$ is reached, our vector of possible estimates grows which is difficult to capture notationally. Since this behavior is irrelevant asymptotically, we directly start with the vector of estimates that would prevail for sample sizes $n\ge N$.} Furthermore, assume that $(X,Y)$ are not perfectly dependent in the sense of Definition \ref{def:perfect_dependence}. Under Assumption \ref{ass:iid}, it holds that
		\begin{enumerate}
			\item We have    
			\begin{align}\label{eq:jointasymptotics1}
            \sqrt{n}\left(
				\begin{pmatrix}
					\vdots\\
					\widehat{\gamma}_n(X^{s_a},Y)\\
					\vdots
				\end{pmatrix}-
				\begin{pmatrix}
					\vdots\\
					\gamma(X^{s_a},Y)\\
					\vdots
				\end{pmatrix}\right)\overset{d}{\rightarrow}(U^{(1)}, ..., U^{(a!)})\sim\mathcal{N}_{a!}(0,\Sigma^1)
			\end{align}
			\item We have    
			\begin{align}\label{eq:jointasymptotics2}
            \sqrt{n}\left(
				\begin{pmatrix}
					\vdots\\
					\widehat{\gamma}_n(X^{s_a},Y^{s_b})\\
					\vdots
				\end{pmatrix}-
				\begin{pmatrix}
					\vdots\\
					\gamma(X^{s_a},Y^{s_b})\\
					\vdots
				\end{pmatrix}\right)\overset{d}{\rightarrow}(V^{(1)}, ..., V^{(a!\cdot b!)})\sim\mathcal{N}_{a!\cdot b!}(0,\Sigma^2).
			\end{align}
			For the asymptotic covariance matrices we have that $\Sigma^1=A^1\Sigma_U^1A^{1\top}$ with an  $a!\times 2\cdot a!$ matrix $A^1$ arising from the Delta Method and a $2\cdot a!\times 2\cdot a!$ matrix $\Sigma_U^1$ which arises from a multivariate CLT for U-statistics. In case 2, we have the same structure with $\Sigma^2=A^2\Sigma_U^2A^{2\top}$ and an $a!\cdot b!\times 2\cdot a!\cdot b!$-dimensional matrix $A^2$ as well as a $2 \cdot a!\cdot b!\times 2\cdot a!\cdot b!$-dimensional matrix $\Sigma_U^2$.
		\end{enumerate}
	\end{proposition}
	Given Proposition \ref{prop:jointasymptotics} (and the precise structure of the involved matrices, which is given in the proof), we can analyze the asymptotic distribution of our estimator. Since the maximum function is not differentiable, we cannot use the standard Delta Method to obtain our result.\footnote{However, a generalized Delta Method by \citet{fang2019inference} is applicable and constitutes an alternative way to obtain the result.} Instead, the asymptotic distribution is non-standard and contains several case distinctions. In case our population exhibits perfect dependence in the sense of Definition \ref{def:perfect_dependence}, the estimator is almost surely equal to 1 (if the sample is large enough such that there is variation in either of the two observed variables). Therefore, we exclude this case in the subsequent proposition.

	\begin{proposition}[Asymptotic Distribution]\label{prop:asymptotics}
		Assume (\ref{eq:jointasymptotics1}) or (\ref{eq:jointasymptotics2}).
		\begin{enumerate}
			\item There exists a unique $\gamma$-maximizing population numbering $s_a^{(*)}$ $(s_b^{(*)})$. In this case, $\widehat{\gamma}^*_n$ inherits the standard asymptotic distribution of $\widehat{\gamma}_n$ \citep{pohle2024inference}. It holds that $$\sqrt{n} \left( \widehat{\gamma}^*_n -\gamma^* \right) \stackrel{d}{\rightarrow} \mathcal{N}(0, \sigma_{\gamma^*}^2) \text{ with }  \sigma_{\gamma^*}^2 =  \frac{1}{(1-\nu^{(*)})^2} \left( \sigma_{\tau^{(*)}}^2 + \gamma^{*2} \sigma_{\nu^{(*)}}^2 + 2 \gamma^* \sigma_{\tau^{(*)} \nu^{(*)}} \right), $$
			where $^{(*)}$ denotes that the respective objects are computed with respect to the numbering $s_a^{(*)}$ $(s_b^{(*)})$. The covariances are in case 1 (see Definition \ref{def:estim}) defined as
			\begin{align*}\sigma_{v^{(*)}w^{(*)}} =&\  4  \E \left[ k_{1}^{(v^{(*)})} (X^{s_a^{(*)}},Y)\, k_{1}^{(w^{(*)})} (X^{s_a^{(*)}},Y) \right] \text{ with } v,w \in \{\tau, \nu\} \text{ and } \sigma^2_{v^{(*)}} := \sigma_{v^{(*)}v^{(*)}}
			\end{align*}
			and 
			\begin{align*}
				k_{1}^{(\nu^{(*)})} (x,y) 
				= \p\Big(X^{s_a^{(*)}}=x\Big) + \p\Big(Y=y\Big) - \p\Big(X^{s_a^{(*)}}=x,Y=y\Big) - \nu^{(*)},\\
				k^{(\tau^{(*)})}_1 (x,y)= 4\, G_{X^{s_a^{(*)}},Y} (x,y)  - 2 \Big(G_{X^{s_a^{(*)}}}(x) + G_Y (y)\Big) + 1 - \tau^{(*)}.
			\end{align*}
            In case 2, we have
            \begin{align*}\sigma_{v^{(*)}w^{(*)}} =&\  4  \E \left[ k_{1}^{(v^{(*)})} (X^{s_a^{(*)}},Y^{s_b^{(*)}})\, k_{1}^{(w^{(*)})} (X^{s_a^{(*)}},Y^{s_b^{(*)}}) \right]
			\end{align*}
			and 
			\begin{align*}
				k_{1}^{(\nu^{(*)})} (x,y) 
				= \p\Big(X^{s_a^{(*)}}=x\Big) + \p\Big(Y^{s_b^{(*)}}=y\Big) - \p\Big(X^{s_a^{(*)}}=x,Y^{s_b^{(*)}}=y\Big) - \nu^{(*)},\\
				k^{(\tau^{(*)})}_1 (x,y)= 4\, G_{X^{s_a^{(*)}},Y^{s_b^{(*)}}} (x,y)  - 2 \Big(G_{X^{s_a^{(*)}}}(x) + G_{Y^{s_b^{(*)}}}(y)\Big) + 1 - \tau^{(*)}.
			\end{align*}
			\item There exist $d$ population numberings $s_a^{(*^1)}, ..., s_b^{(*^d)}$ or $(s_a^{(*^1)}, s_b^{(*^1)}), ..., (s_a^{(*^d)}, s_b^{(*^d)})$ $(d\in \mathbb{N}_{\ge 2})$ with $*^1, ..., *^d \in \{1, ..., a!\}$ or $*^1, ..., *^d \in \{1, ..., a!\cdot b!\}$ that are solutions to the maximization problem in Definition \ref{def:wermuth_coefficient}. In case 1, it holds that $$\sqrt{n} \left( \widehat{\gamma}^*_n -\gamma^* \right) \stackrel{d}{\rightarrow} \max\{U^{(*^1)}, ..., U^{(*^d)}\}$$
			with $U^{(*^1)}, ..., U^{(*^d)}$ defined as in Proposition \ref{prop:jointasymptotics}. Similarly in case 2, we have $$\sqrt{n} \left( \widehat{\gamma}^*_n -\gamma^* \right) \stackrel{d}{\rightarrow} \max\{V^{(*^1)}, ..., V^{(*^d)}\}$$
			with $V^{(*^1)}, ..., V^{(*^d)}$ defined as in Proposition \ref{prop:jointasymptotics}.
		\end{enumerate}
	\end{proposition}
 \begin{rem}\label{rem:casedistinctions}
     The number of case distinctions implicit in Proposition \ref{prop:asymptotics} is very large. It amounts to
     \begin{align*}
         \sum_{d=1}^{a!}\binom{a!}{d} \quad \text{or} \quad \sum_{d=1}^{a!\cdot b!}\binom{a!\cdot b!}{d}
     \end{align*}
     in the respective cases of one or two nominal random variables.
 \end{rem}
 \subsection{Confidence Intervals}
 Constructing confidence intervals for the new correlation coefficient is complicated due to the non-standard asymptotic distribution described in Proposition \ref{prop:asymptotics}. Although a complicated distribution could still be exploited (see, e.g., the confidence intervals for Cole's C \citep{cole1949} in \citet{pohle2024measuring} constructed by an inverted testing procedure), the problem is that the researcher does not know which of the many possible asymptotic distributions prevails in her particular setting. 
 
 Facing a similar problem, \citet{holzmann2024lancaster} propose several methods to construct confidence intervals using standard procedures or bootstrap. For some of those methods, they also take into account the case distinctions they have in their asymptotic distribution. I refrain from such an involved approach for the following reasons. Firstly, I do not consider bootstrap because due to the rapidly increasing nature of the factorial, estimating the coefficient is already computationally intensive enough (see also footnote \ref{foot:computationalintensive}). Secondly, the number of case distinctions in our asymptotic distribution is so large that I do not see any promising path to combine them in an attempt to construct (possibly conservative) confidence intervals (see Remark \ref{rem:casedistinctions}). Instead, we make a simplifying assumption.
 \begin{assumption}\label{ass:simplifying}
     There exists a unique $\gamma$-maximizing population numbering $s_a^{(*)}$ $(s_b^{(*)})$.
 \end{assumption}
 With this assumption (and the assumption that $(X,Y)$ are not perfectly dependent), we are in the well-behaved first case of Proposition \ref{prop:asymptotics} and can construct our confidence intervals with asymptotic coverage probability $1-\alpha$ as
 \begin{align*}
     [\max\{\widehat{\gamma}^*_n-z_{1-\alpha/2}\widehat{\sigma}_{\gamma^*}/\sqrt{n},0\}, \min\{\widehat{\gamma}^*_n+z_{1-\alpha/2}\widehat{\sigma}_{\gamma^*}/\sqrt{n},1\}],
 \end{align*}
 where $z_{1-\alpha/2}=\Phi^{-1}(1-\alpha/2)$ and $\widehat{\sigma}_{\gamma^*}$ is equal to $\widehat{\sigma}_{\gamma}$ as defined in \citet{pohle2024inference} relative to the numbering that $\widehat{\gamma}^*_n$ has chosen. 
 \begin{lemma}\label{lem:variance_estimator_consistency}
 Under Assumptions \ref{ass:iid} and \ref{ass:simplifying},
     $\widehat{\sigma}^2_{\gamma^*} \overset{p}{\to} \sigma^2_{\gamma^*}$.
 \end{lemma}
 The simulation section shows that this approach delivers good coverage probabilities for a set of reasonable DGPs.
 \subsection{Independence Test}
Independence tests for two discrete (or even nominal) random variables are readily available in the literature and have been in scientific use for decades. Notable examples are the Chi-square test of independence and the G-Test \citep{mcdonald2014handbook}. However, if we consider a nominal variable together with a continuous or discrete-continuous variable, testing options are quite limited. One possibility would be to discretize the continuous part and revert to the G-Test, but the arbitrariness in choosing the width and number of bins is suboptimal. Another option would be to regress the continuous variable on dummy variables that jointly determine the value of the nominal variable and subsequently perform a global F-Test on all the slope coefficients.
\begin{lemma}\label{lem:globalFTest}
    Let $X:\Omega \to \Omega'$ and $Y:\Omega \to \mathbb{R}$ with $\Omega':=\{\omega_1', ..., \omega_a'\}$. Consider the regression model $$Y=b_1+\sum_{j=2}^{a}b_j\mathds{1}_{\{X=\omega_j'\}}+U \ a.s.$$
    with $U:\Omega \to \mathbb{R}$ independent from $X$. It holds that $$X\perp\!\!\!\perp Y \Leftrightarrow b_j=0\ \forall j\in\{2, ..., a\}.$$
\end{lemma}

If the real-valued variable is discrete-continuous, one could in principle extend this strategy to corner solution models such as Tobit.

Since we know the asymptotic distribution of our correlation coefficient under independence, we can add a further option based on this coefficient. Unlike as in the construction of confidence intervals, Assumption \ref{ass:simplifying} is clearly inappropriate. Therefore, the asymptotic distribution is more complicated.\footnote{For an analytical expression of its pdf, see Corollary 4 in \citet{arellano2008exact}. Lemma 4 in \citet{holzmann2024lancaster} also gives the cdf, albeit only in the bivariate case.}
\begin{lemma}\label{lem:gamma*asymptotics_independence}
    Suppose $(X,Y)$ are independent and either only $X$ (case 1) or $X$ and $Y$ (case 2) have a nominal scale. In case 1, it holds that $$\sqrt{n}  \widehat{\gamma}^*_n \stackrel{d}{\rightarrow} \max\{U^{(1)}, ..., U^{(a!)}\}$$
			with $U^{(1)}, ..., U^{(a!)}$ defined as in Proposition \ref{prop:jointasymptotics}. Similarly in case 2, we have $$\sqrt{n}  \widehat{\gamma}^*_n  \stackrel{d}{\rightarrow} \max\{V^{(1)}, ..., V^{(a!\cdot b!)}\}$$
			with $V^{(1)}, ..., V^{(a!\cdot b!)}$ defined as in Proposition \ref{prop:jointasymptotics}.
\end{lemma}
In order to exploit this lemma for an independence test, we need to be able to compute p-values via the cdf of the limiting distribution. Since for arbitrary $\mathbb{R}$-valued random variables $W_1, ..., W_n$, it holds that
\begin{align*}
    \p(\max\{W_1, ..., W_n\}\le w)=\p(W_1\le w, ..., W_n\le w),
\end{align*}
this can easily be done via the results obtained in Proposition \ref{prop:jointasymptotics} and an \texttt{R} package that is able to evaluate the CDF of a multivariate normal distribution, for example \texttt{mvtnorm} \citep{mvtnorm}\footnote{If $a$ and $b$ are large (see Appendix \ref{sec-App:Computation} for precise values), a very high-dimensional normal CDF has to be evaluated. At least in case 2, it is then advisable to use a computationally less challenging independence test, of which plenty exist.}. However, we need to estimate the matrices $A^1$ and $\Sigma_U^1$ (case 1) or $A^2$ and $\Sigma_U^2$ (case 2). For the $A$-matrices, we just replace the several $\tau$ and $\nu$ by their empirical counterparts and possibly reduce the size of the matrix if not all values with positive probability are present in the sample. For the $\Sigma_U$-matrices, we estimate $\sigma_{v(i)w(j)}$ with $v, w\in\{\tau, \nu\}$ and $i,j\in \{1, ..., k!\}$ or $i,j\in \{1, ..., k!\cdot l!\}$ in the spirit of \citet{pohle2024inference}, where $k,l$ are defined as in Definition \ref{def:estim}.
	
	\section{Simulations}\label{sec:sims}
	As a means to evaluate the finite sample performance of the introduced confidence intervals and the independence test, this section presents simulated coverage rates as well as size and power values for a collection of DGPs. Additionally, it shows simulated mean bias values to showcase the finite sample performance of the proposed estimator. Each simulation contains $MC=1,000$ pairs of iid observations, features several degrees of dependence and different sample sizes.
	\subsection{DGPs}\label{subsec:DGPs}
	\subsubsection{Nominal-Continuous DGPs}
	The first two DGPs are inspired by classical linear regression in which nominal explanatory variables can be accounted for by including binary indicators. $X$ is defined via $\mathbb{P}(X=A)=\mathbb{P}(X=B)=\mathbb{P}(X=C)=1/3$ and $Y$ via 
	\begin{align*}
		Y=\alpha \mathds{1}_{\{X=B\}}-\alpha \mathds{1}_{\{X=C\}}+U, \quad U\sim \mathcal{N}(0,1), t(1), \quad \alpha \in \mathbb{R}.
	\end{align*}
	The second pair of DGPs is based on the multinomial logit model with $X\sim \mathcal{N}(0,1), t(1)$ and \begin{align*}\mathbb{P}(Y=A)&=\exp(-\alpha X) / (1 + \exp(-\alpha X) + \exp(\alpha X)),\\
	\mathbb{P}(Y=B)&=\exp(\alpha X) / (1 + \exp(-\alpha X) + \exp(\alpha X)),\\
	\mathbb{P}(Y=C)&=1 / (1 + \exp(-\alpha X) + \exp(\alpha X)), \quad \alpha \in \mathbb{R}.
	\end{align*}
    For both DGPs, $\alpha = 0$ leads to independent $X$ and $Y$ and the larger $\alpha$ gets in absolute value, the stronger the dependence\footnote{\label{foot:stronger_dep}In Section \ref{sec:desirable_properties}, I admit that I do not have any concept of dependence ordering in the nominal case. Therefore, the word ``stronger'' shall be understood in a heuristic sense.}.
	
	\subsubsection{Nominal-Nominal DGPs}
    Table \ref{tab:DGPsnomnom} shows two DGPs in which both variables have a nominal scale. The first has one uniform and one skewed marginal distribution while the second features two uniform distributions. Again, $\alpha = 0$ corresponds to independence and the farther $\alpha$ is away from 0, the stronger the dependence\textsuperscript{\ref{foot:stronger_dep}}.
    
    \begin{table}[tb]
        \centering
        \caption{Two DGPs in which either of the variables has a nominal scale.}
            \label{tab:DGPsnomnom}
        \begin{tabular}{c}
        \begin{subtable}{\textwidth}
			\centering
            \subcaption{DGP with one uniform and one skewed distribution.}
            \label{tab:DGPunifskew}
			\begin{tabular}{ccccc}
				\toprule
				& $y_1$ &  $y_2$ & $y_3$ \\
				\midrule 
				$x_1$&$76/300+(2/30)\alpha$&$4/100-(1/30)\alpha$&$4/100-(1/30)\alpha$&$1/3$\\
				$x_2$&$76/300+(2/30)\alpha$&$4/100-(1/30)\alpha$&$4/100-(1/30)\alpha$&$1/3$\\  
				$x_3$&$76/300-(4/30)\alpha$&$4/100+(2/30)\alpha$&$4/100+(2/30)\alpha$&$1/3$\\
				&$76/100$&$12/100$&$12/100$&1\\
				\bottomrule
			\end{tabular}
            \end{subtable}\\
            
            \begin{subtable}{\textwidth}
			\centering
            \subcaption{DGP with two uniform distributions.}
            \label{tab:DGPunifunif}
			\begin{tabular}{ccccc}
				\toprule
				& $y_1$ &  $y_2$ & $y_3$ \\
				\midrule 
				$x_1$&$1/9+(2/30)\alpha$&$1/9-(1/30)\alpha$&$1/9-(1/30)\alpha$&$1/3$\\
				$x_2$&$1/9+(2/30)\alpha$&$1/9-(1/30)\alpha$&$1/9-(1/30)\alpha$&$1/3$\\  
				$x_3$&$1/9-(4/30)\alpha$&$1/9+(2/30)\alpha$&$1/9+(2/30)\alpha$&$1/3$\\
				&$1/3$&$1/3$&$1/3$&1\\
				\bottomrule
			\end{tabular}
            \end{subtable}
            \end{tabular}
		\end{table}
	\subsection{Results}
    This subsection presents the simulation results for the confidence intervals and the estimator by means of empirical coverage rate and mean bias graphs and the simulation results for the independence test by means of p-value histograms and power tables. 
  
    \begin{figure}
    	\begin{subfigure}{0.49\textwidth}
    		\includegraphics[width=\linewidth]{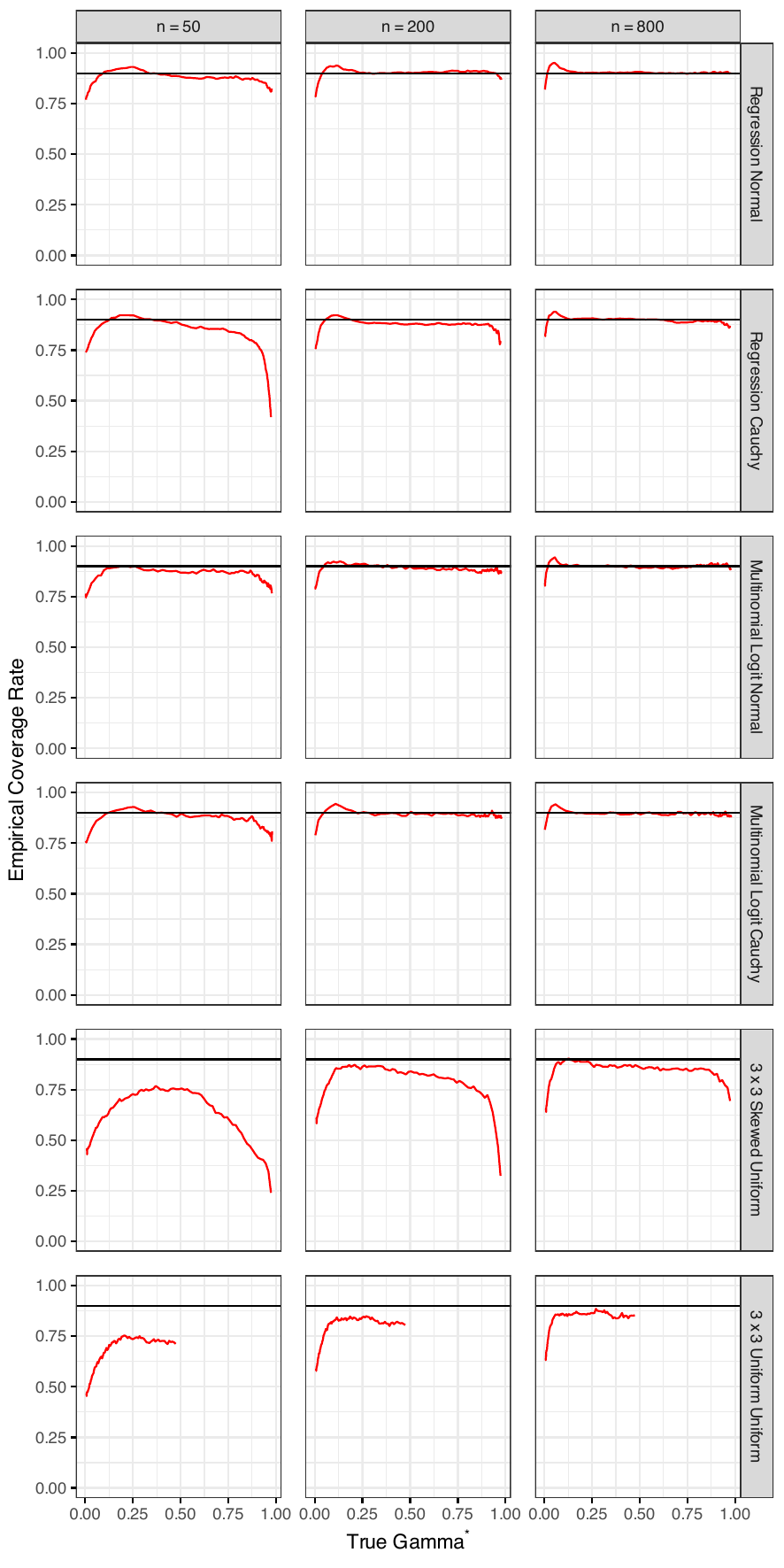}
    		\caption{Confidence interval coverage}
    		\label{fig:coverage}
    	\end{subfigure}
    	\hfill
    	\begin{subfigure}{0.49\textwidth}
    		\includegraphics[width=\linewidth]{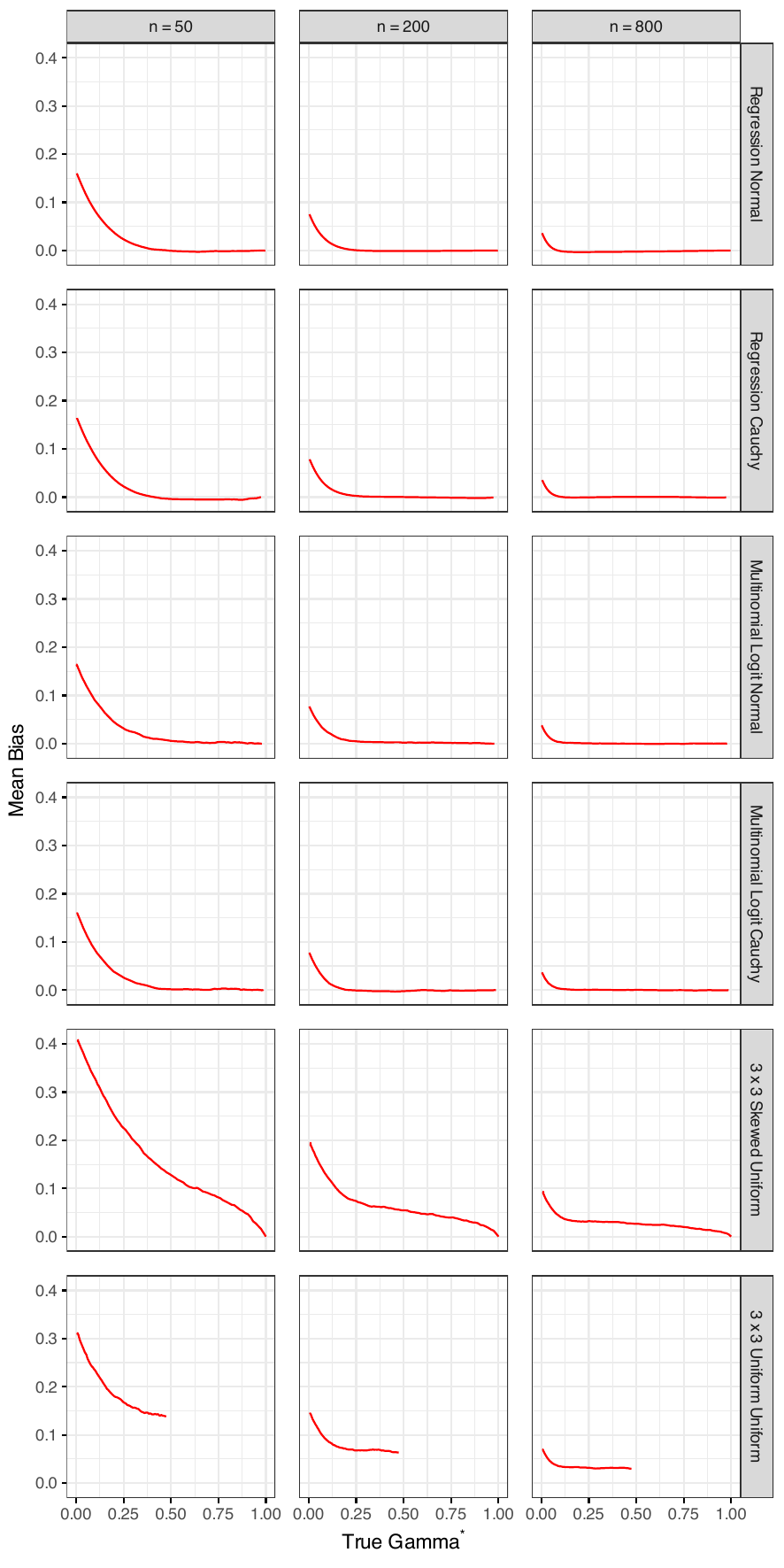}
    		\caption{Mean bias}
    		\label{fig:bias}
    	\end{subfigure}
        \caption{This figure shows empirical confidence interval coverage rates (nominal level: $0.9$) and mean bias plots for $\gamma^*$ with $MC = 1,000$ simulation replications. The sample sizes $n\in \{50, 200, 800\}$ are given in the column and the DGPs are given in the row descriptions, respectively. The x-axis displays the true values of $\gamma^*$ that arise for a specific choice of $\alpha$.}
    \end{figure}
    For the empirical coverage rates, we desire a value close to the nominal value of $90\%$. Although Assumption \ref{ass:simplifying} is clearly not satisfied for the nominal-continuous DGPs in the independence case $\gamma^*=0$, the coverage rates in Figure \ref{fig:coverage} still appear acceptable, especially for stronger dependencies. The nominal-nominal DGPs are deliberately designed such that Assumption \ref{ass:simplifying} is never satisfied (note that $y_2/y_3$ and $x_1/x_2$ occur with the same marginal and joint probabilities irrespective of the choice of $\alpha$) and this feature is reflected by low coverage probabilities, especially for $n=50$. For larger sample sizes, the empirical coverage rate increases almost to the nominal level. Interestingly, the mean bias plots in Figure \ref{fig:bias} show a similar behavior as the empirical coverage rate plots. The bias is large and positive under independence and decreases both with the sample size and the degree of dependence. That is because the bias arises also mainly due to violations of Assumpotion \ref{ass:simplifying}. Intuitively, the bias increases with the intensity of permutation switching of the estimator across the Monte Carlo replications. If Assumption \ref{ass:simplifying} was satisfied robustly, i.e.\ in the sense that also in finite samples, the estimator (almost) always chooses the same numbering, there would be no bias.

    Regarding the independence test, recall that under the null hypothesis $(\alpha = 0)$, the histograms shall converge to a continuous uniform distribution on $[0,1]$. For $\alpha \ne 0$, we desire a high rejection rate of the null hypothesis, i.e.\ high power.
    
    Figure \ref{fig:gamma*_a0} displays p-value histograms for all the considered DGPs under independence where the test has been conducted based on $\gamma^*$. They all show convergence towards the desired distribution, albeit at different rates. Figure \ref{fig:gamma*_a0_reg} in turn shows the results if we change our testing procedure to traditional approaches (global F-Test and $\chi^2$ test\footnote{The G-Test shows worse results than the $\chi^2$ test and is therefore omitted.}). The uniformity of our p-value histograms no longer holds for those DGPs in which the continuous variable follows a Cauchy distribution. This illustrates the major strength of the new testing approach: Due to the rank-based nature of the new measure, it does not rely on any moment conditions. 

    Comparing the power values of the different tests in Table \ref{tab:power}, we find major advantages of the $\gamma^*$ independence test in the two Cauchy cases. However, there are also large differences in one DGP for which no assumption that the traditional tests make is violated. The DGP defined in Table \ref{tab:DGPunifskew} delivers a remarkable increase in power if we switch from the $\chi^2$ independence test to the $\gamma^*$ independence test. The reason for that may be the slow convergence of the $\chi^2$ test statistic to its limiting distribution if some cells in the contingency table have very small probabilities \citep{bruce2015}. For the remaining DGPs, the traditional tests have usually more power.
    \begin{figure}
    \begin{subfigure}{0.49\textwidth}
        \includegraphics[width=\linewidth]{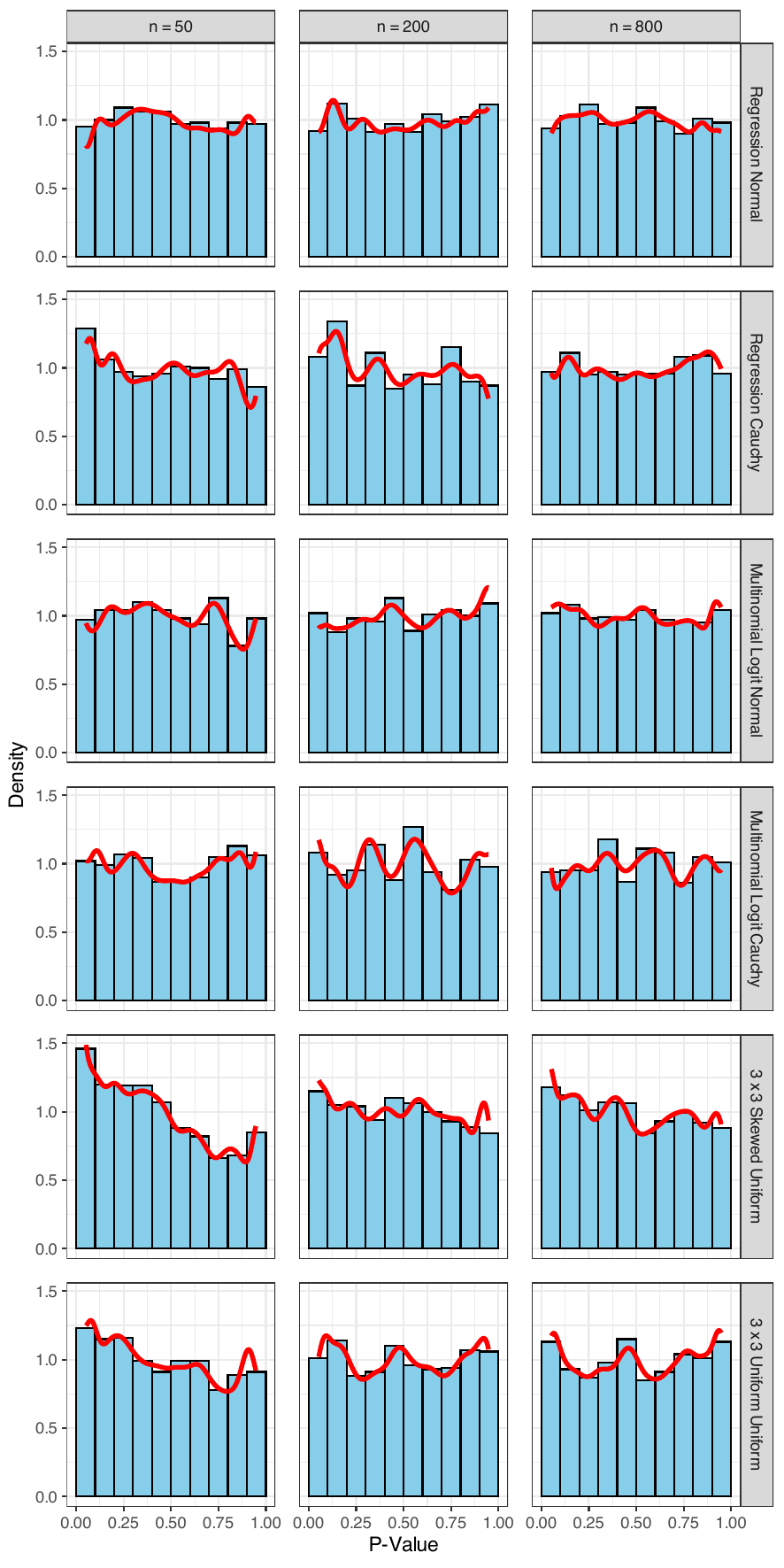}
        \caption{$\gamma^*$ independence test}
        \label{fig:gamma*_a0}
    \end{subfigure}
    \hfill
    \begin{subfigure}{0.49\textwidth}
        \includegraphics[width=\linewidth]{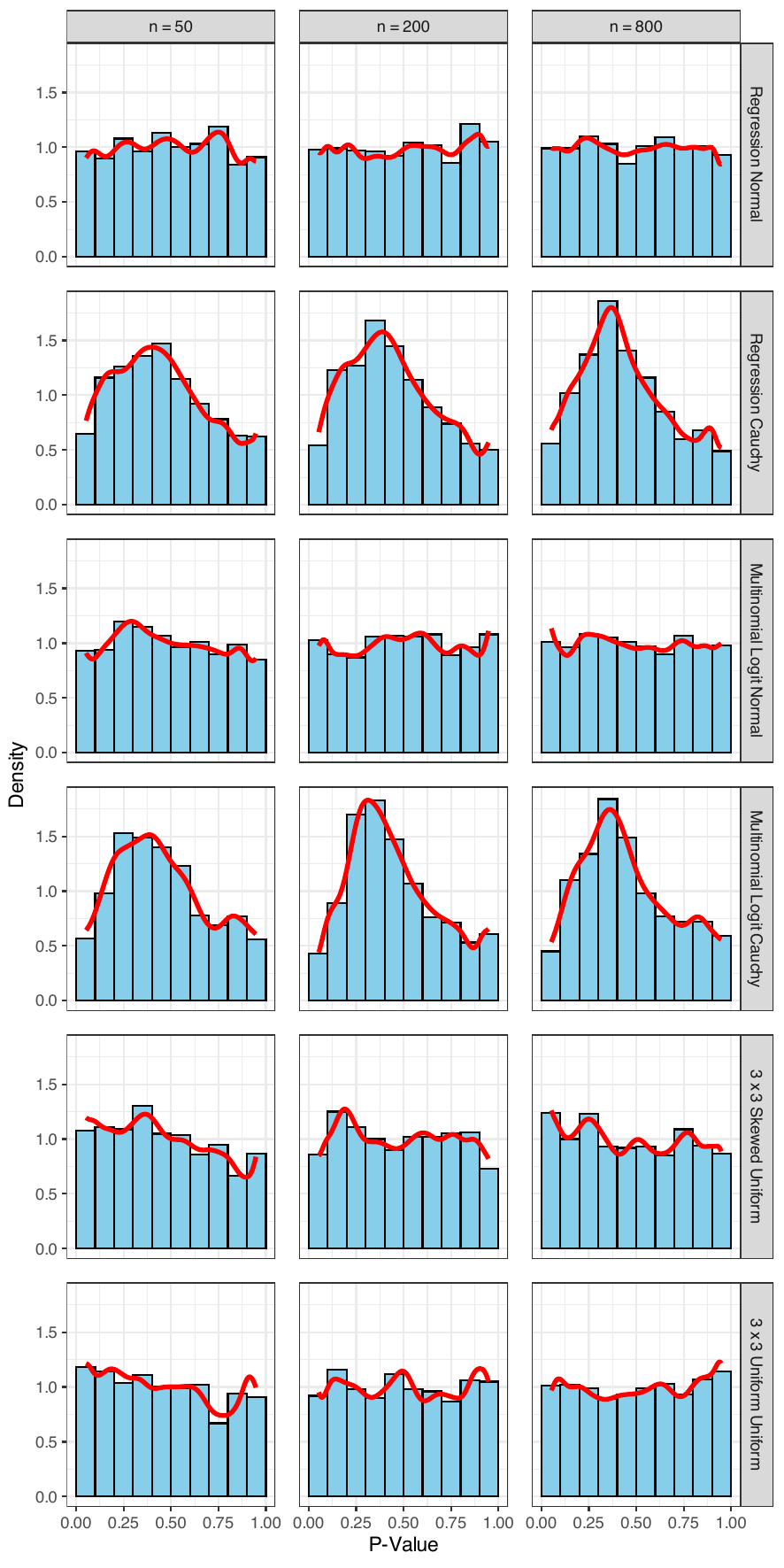}
        \caption{global F-test or $\chi^2$ independence test}
        \label{fig:gamma*_a0_reg}
    \end{subfigure}
    \caption{This figure shows p-value histograms that arise either from the newly introduced testing scheme (left subplot) or from traditional testing approaches (right subplot). It considers three different sample sizes and six different DGPs which are explained in Subsection \ref{subsec:DGPs}. For all DGPs, sampling happens under independence $(\alpha = 0)$ with 1,000 Monte Carlo repetitions.}
    \end{figure}

    \begin{table}[tb]                
    \centering
    \caption{Power values for the $\gamma^*$ independence test and the traditional tests (global F-test and $\chi^2$ independence test) at 10\% significance with 1,000 Monte Carlo repetitions for all DGPs described in Subsection \ref{subsec:DGPs}. The following abbreviations are used: Regression Normal (RN), Regression Cauchy (RC), Multinomial Logit Normal (MLN), Multinomial Logit Cauchy (MLC), 3x3 Skewed Uniform (3x3 SU), 3x3 Uniform Uniform (3x3 UU).}
            \label{tab:power}  
        	\begin{tabular}{ccccccc}
			\addlinespace
			\toprule
			$\gamma^*$&\multicolumn{5}{c}{$\gamma^*=0.1$}\\			
            \midrule 
            & RN & RC & MLN & MLC & 3x3 SU & 3x3 UU\\
			$n=50$&0.153&0.174&0.186&0.177&0.174&0.156\\
			$n=200$&0.356&0.395&0.391&0.338&0.175&0.26\\
			$n=800$&0.831&0.857&0.859&0.859&0.43&0.718 \\
			\midrule
            F, $\chi^2$&\multicolumn{5}{c}{$\gamma^*=0.1$}\\			
            \midrule 
            & RN & RC & MLN & MLC & 3x3 SU & 3x3 UU\\
			$n=50$&0.152&0.076&0.185&0.126&0.131&0.153\\
			$n=200$&0.358&0.063&0.400&0.505&0.165&0.261\\
			$n=800$&0.860&0.053&0.878&0.832&0.395&0.729\\
            \bottomrule
		\end{tabular}
\end{table}

    \section{Application}\label{sec:case_study}
    This section illustrates the two major innovations of this paper in two data examples. 
    \subsection{Case Study on Countries and Income}
    Firstly, the proposed coefficient is the first coefficient which is able to summarize the dependence within a bivariate random vector if one marginal distribution is continuously distributed and the other has a nominal scaling. A self-evident example for such a situation is the fundamental observation of economics: People in some countries have higher incomes than people in others. 
    
    Figure \ref{fig:CountryIncome} displays the dependence between the variables country and income measured by the new coefficient. Each colored point is located at a border triangle and represents the value of the coefficient for those three countries. That is, the variable country always takes exactly three values and the income data is chosen accordingly. The same information with the corresponding confidence intervals is displayed in Figure \ref{fig:CountryIncome_CIs}, where the x-axis shows the border triangle (see the country abbreviations in Table \ref{tab:samplesizes}). It is clearly visible that border triangles featuring countries like France or Poland that have large sample sizes exhibit narrower confidence intervals than other border triangles, for example the one between Hungary, Austria and Slovenia.

    With respect to the sizes of the coefficients, those comparisons in which two countries with a similar mean income (see Table \ref{tab:samplesizes}) such as France and Germany are a part of exhibit rather small values. In contrast to that, the Republic of Serbia, Hungary and Romania have a clear income order which is reflected by a larger coefficient. However, other comparisons like Slovakia, the Czech Republic and Austria also have an unambiguous income order in terms of the mean but still yield a small coefficient estimate. The reason for this behavior is of course that even heavily differing means can be explained by few outliers, a differing number of zeros and other factors that have little influence on a rank-based measure like $\widehat{\gamma}^*$.
    \begin{figure}
    \begin{subfigure}{0.49\textwidth}
        \includegraphics[width=\linewidth]{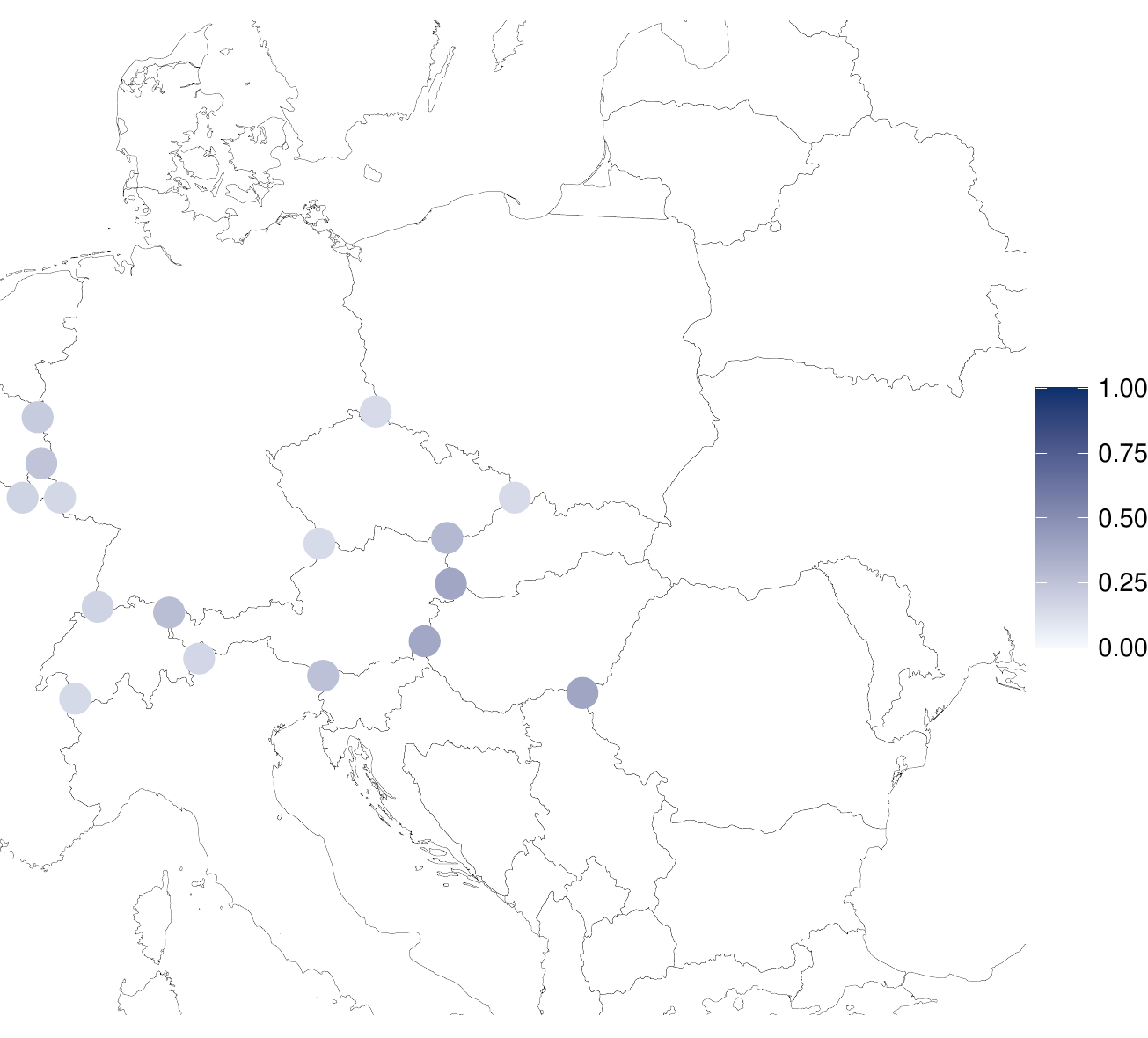}
        \caption{without confidence intervals}
        \label{fig:CountryIncome}
    \end{subfigure}
    \hfill
    \begin{subfigure}{0.49\textwidth}
        \includegraphics[width=0.9\linewidth]{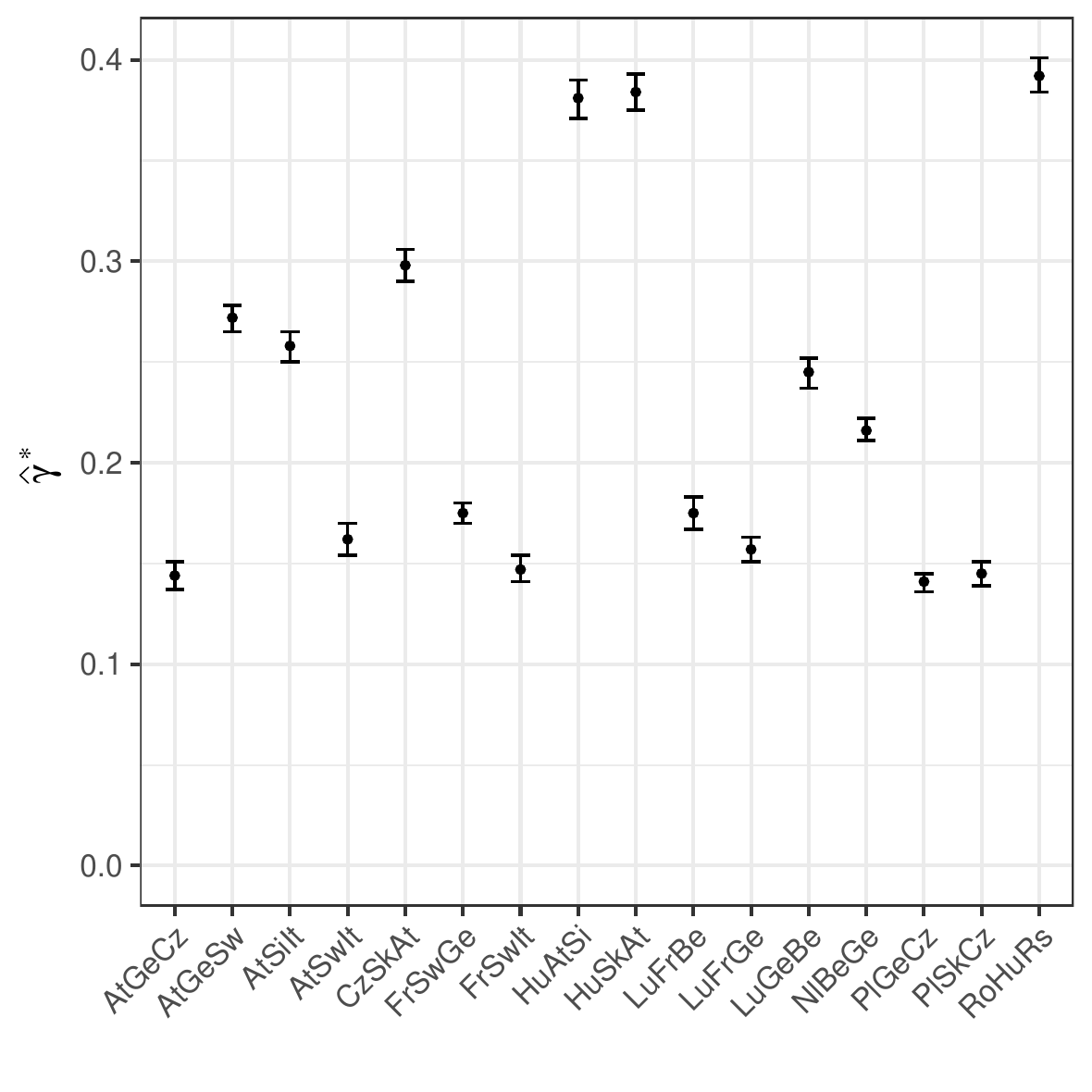}
        \caption{with confidence intervals}
        \label{fig:CountryIncome_CIs}
    \end{subfigure}
    \caption{This figure shows $\widehat{\gamma}^*$ for the comparison between the variables country and income chosen from the Luxembourg Income Study (LIS) Database (2011 PPP).}
    \end{figure}
    \begin{table}[tb]
	\centering
\caption{This table lists all countries that are considered in this case study together with their abbreviations, their respective numbers of income replies and their respective mean income values.}
	\label{tab:samplesizes}
    \begin{threeparttable}
		\begin{tabular}{cccc}
			\toprule
			\textbf{abbreviation} & \textbf{country}  & \textbf{sample size}&\textbf{mean income} \\
			\midrule
			At & Austria &12,096 & 32,597.75\\
			Be & Belgium & 15,030 &28,415.33\\
			Cz & Czech Republic& 19,205& 14,264.06\\
			Fr & France & 86,392& 23,135.25\\
			Ge & Germany & 50,563& 24,230.60\\
			Hu & Hungary& 6,237& 8,317.53\\
			It & Italy& 15,198& 35,588.65 \\
			Lu & Luxembourg& 9,098& 46,569.78\\
			Nl & Netherlands& 30,378& 32,691.95\\
			Pl & Poland & 87,603& 10,128.91\\
			Ro & Romania & 16,609& 13,997.20\\
			Rs & Republic of Serbia& 15,309& 5,110.57\\
			Si & Slovenia & 11,228& 11,824.54\\
			Sk & Slovakia & 14,653& 10,548.63\\
            Sw & Switzerland & 18,215& 38,723.87\\
			\bottomrule
		\end{tabular}
	\end{threeparttable}
\end{table}
    \nocite{LIS2025}
	\subsection{Case Study on Countries and Religions}
    Secondly, the proposed coefficient satisfies the property of attainability. As described in Section \ref{sec:improper_measures}, the traditional measures only satisfy this property if the marginal distributions have a certain structure. In the balanced setting of a $3\times 3$ contingency table, this amounts to the requirement that the marginals are identical in terms of the probability mass. That means that the values of the variables can differ but for each value of the first variable there needs to exist a value of the second variable with the same marginal probability mass. 
    
     An example for a comparison in which this requirement is particularly heavily violated is the one between the variables country and religion. Population data on those variables is available in the world religion database
     \citep{WRD2025}.
     Since we only consider the three monotheistic world religions Christianity, Islam and Judaism and perform the analysis for each border triangle separately, we are indeed in the setting of $3\times 3$ contingency tables. However, the marginals in those tables are typically very different. For example, European countries usually have Christian population majorities and very little Muslims and even fewer Jews. As a consequence, the marginal distribution of the variable religion is skewed towards Christianity but the marginal distribution of the variable country may be more balanced.

    \begin{figure}
        \centering
        \includegraphics[width=\textwidth]{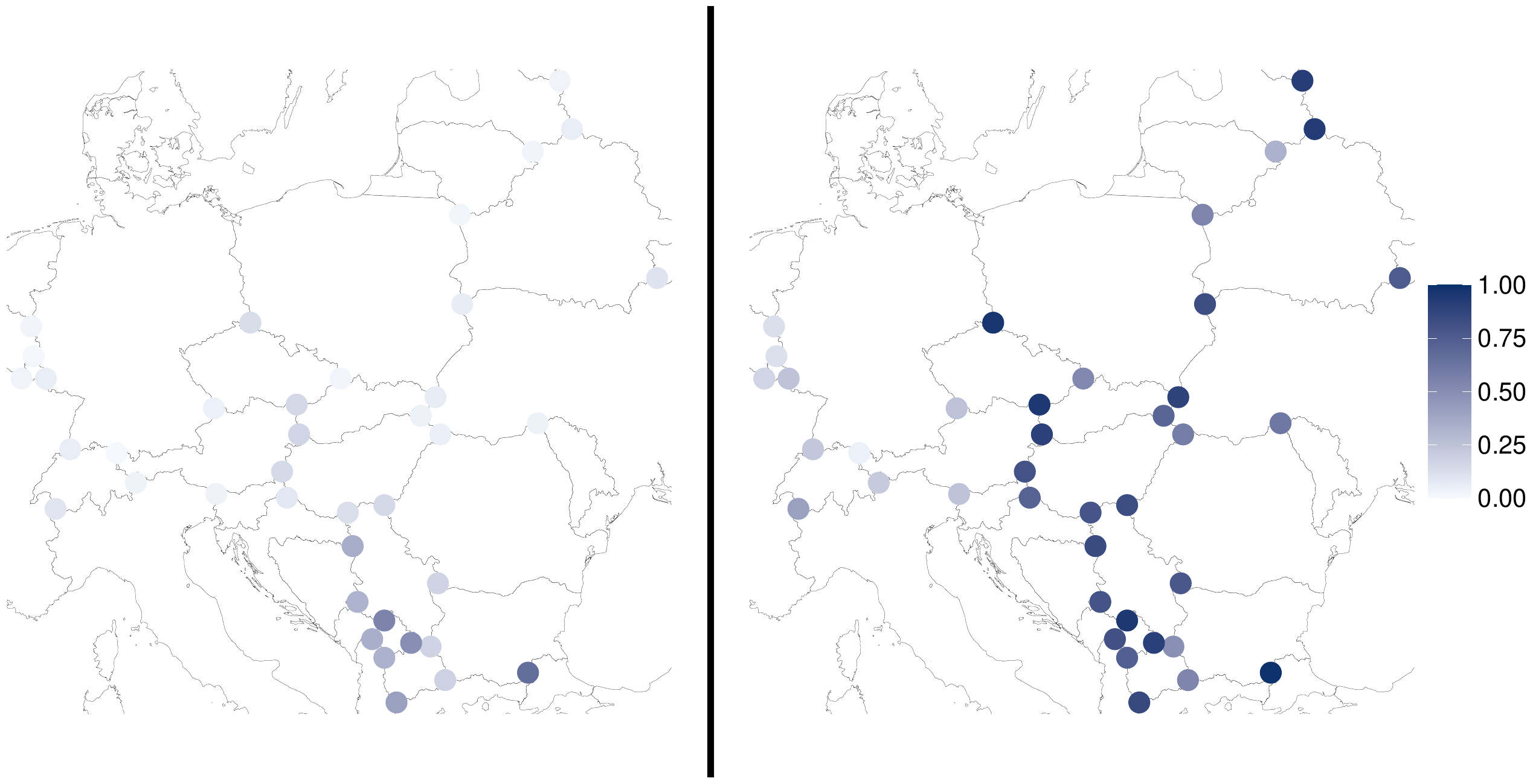}
        \caption{This figure shows Cramér's $V$ (left plot) and $\gamma^*$ (right plot) for the dependence between the variables country and religion. Each colored dot represents the value of the respective coefficient if the variable country is restricted to the three countries featuring in the border triangle at which the dot is located. The variable religion is always restricted to the three monotheistic world religions Christianity, Islam and Judaism.}
        \label{fig:Rel_CramerWermuth}
    \end{figure}
    
     The differing assessments of dependence between the new coefficient and the old coefficients, where Cramér's V serves as a representative for the traditional measures, are now nicely illustrated in Figure \ref{fig:Rel_CramerWermuth}. The maps for the remaining measures are located in Appendix \ref{sec-App:addGraphsTables} (Figure \ref{fig:Rel_PC_tau_lambda_uncert}), but display qualitatively similar correlation patterns to the one for Cramér's V. Two observations are immediate: On the one hand, the new coefficient generally yields larger values. On the other hand, the difference varies heavily across the comparisons. One particularly good example for the strength of the new coefficient are the border triangles Germany, Austria, Switzerland and Germany, Poland, Czechia. Cramér's V indicates a rather weak dependence for both triangles, although the value for the latter triangle is a little larger. The newly proposed coefficient however yields almost the same value as Cramér's V for the German-speaking triangle and a value close to one for the other triangle. Indeed, the dependence in both comparisons is very different. Germany, Austria and Switzerland have very similar social structures and religious groups of similar (relative) size. Thus, the dependence between country and religion is weak in this example. But in contrast to Germany, Poland and Czechia have enforced very restrictive migration policies in the past, which results in small religious minorities, especially on the Muslim side. Because both countries also barely have Jews, the Muslims and the Jews are in this example almost fully concentrated in Germany. The reason why Cramér's V yields such a small value is therefore not the weak dependence but the extremely skewed marginal distribution of the variable religion which is not matched by an equally skewed marginal distribution of the variable country. An additional reason is the structurally different definition of perfect dependence that I defend in this paper which can lead to large differences between traditional measures and $\gamma^*$ even if the marginal probabilities were such that the traditional measures are attainable (see Table \ref{tab:MSC_DC_counterexample}).
    
     Another nice example is the border triangle Bulgaria, Greece and Turkey. While Bulgaria and Greece have Christian majorities, Turkey is a primarily Muslim country. This situation of an obvious statistical dependence which consists of large religious blocks concentrated in different countries is nicely captured by both coefficients.

	\section{Conclusion}
	\label{sec:conclusion}
	This paper has proposed a new concept of perfect dependence between two random variables of which at least one has a nominal scale that is attainable irrespective of the involved marginal distributions. In addition to that, it introduces the notion of \textit{proper} dependence measures to the nominal case by defining a set of desirable properties which each such measure has to fulfill. All existing dependence measures fail to satisfy some of the axioms, especially since the existence axiom requires the ability to deal with a nominal-continuous combination of marginal distributions and the attainability axiom requires compliance with the previously defined concept of perfect dependence. As a result, new measures are necessary which are also proposed in this paper. For one of those, I develop asymptotic theory that can  subsequently be exploited for confidence intervals and independence testing. Simulations show good finite sample performance for both methods of statistical inference and two applications illustrate the superiority of the new coefficient relative to existing ones.
	
	Future research could continue the quest for a proper dependence measure also in the case in which only one nominal random variable is present. Also, it would be interesting to extend the notion of margin-freeness to nominal random variables and include this feature into the set of desirable properties \citep{geenens2020copula}.
    \bibliography{literature_autocorrelations_ordinal}

\clearpage
    \appendix
    \appendixpage
\numberwithin{equation}{section}
\numberwithin{table}{section}
\numberwithin{figure}{section}
\numberwithin{theorem}{section}
\section{Further Improper Dependence Measures}\label{sec:AppendixImproperMeasures}
	\subsection{Measures Based on Proportional Reduction in Predictive Error}
	The second most popular class of dependence measures has been proposed in the first of a series of seminal papers by \citet{Goodman1954}. It consists of Goodman-Kruskal's $\lambda$ and $\tau$. Again, both measures only work in a setting of two discrete random variables. Introducing the notations \begin{align*}
		p_{m\cdot}:= \max_{i}p_{i\cdot},\quad p_{\cdot m}:= \max_{j}p_{\cdot j}, \quad  p_{mj}:=\max_{i}p_{ij}\quad \text{and} \quad p_{im}:=\max_{j}p_{ij},
	\end{align*}
	one can define 
	\begin{align*}
		\lambda_y(X,Y):=\frac{\sum_{i=1}^{a}p_{im}-p_{\cdot m}}{1-p_{\cdot m}} \quad \text{and} \quad \lambda_x(X,Y):=\frac{\sum_{j=1}^{b}p_{mj}-p_{m \cdot}}{1-p_{m \cdot}}.
	\end{align*}
	The interpretation of $\lambda_y$ is the following: If someone (with knowledge about the distribution of $(X,Y)$) has to guess the category of $Y$ in which an unknown realization of $(X,Y)$ is in, the best guess would be to choose the category with the highest probability. This results in a misclassification probability of $1-p_{\cdot m}$. If on the other hand the value of $X$ that this unknown realization takes was known (and e.g.\ $x_i$), the probability of misclassification would be $1-p_{im}/p_{i\cdot}$. However, since this is a conditional probability we can use the law of total probability to obtain
	\begin{align}\label{eq:LTP}
		\mathbb{P}(Error)=\sum_{i=1}^{a}\mathbb{P}(Error|X=x_i)\mathbb{P}(X=x_i)=\sum_{i=1}^{a}(1-p_{im}/p_{i\cdot})p_{i\cdot}=1-\sum_{i=1}^{a}p_{im}.
	\end{align}
	$\lambda_y$ is now the relative reduction in error when guessing the category of $Y$ with knowing $X$ relative to without knowing $X$. The interpretation of $\lambda_x$ is analogous.
	
	Since both measures are asymmetric, one can construct a symmetrized measure as
	\begin{align*}
		\lambda(X,Y):=\frac{\frac{1}{2}(\sum_{i=1}^{a}p_{im}+\sum_{j=1}^{b}p_{mj}-p_{\cdot m}-p_{m \cdot})}{1-\frac{1}{2}(p_{\cdot m}+p_{m\cdot})}.
	\end{align*}
	The interpretation is again a reduction in predictive error but this time with a $50\%$ probability that the category of $X$ and a $50\%$ probability that the category of $Y$ has to be guessed.
	
	Goodman-Kruskal's $\tau$ in turn can be defined as
	\begin{align*}\tau_y(X,Y):=\frac{\sum_{i=1}^{a}\sum_{j=1}^{b}p_{ij}^2/p_{i\cdot}-\sum_{j=1}^{b}p_{\cdot j}^2}{1-\sum_{j=1}^{b}p_{\cdot j}^2} \quad \text{and} \quad \tau_x(X,Y):=\frac{\sum_{i=1}^{a}\sum_{j=1}^{b}p_{ij}^2/p_{\cdot j}-\sum_{i=1}^{a}p_{i\cdot}^2}{1-\sum_{i=1}^{a}p_{i\cdot}^2}.
	\end{align*}
	It has a similar interpretation as $\lambda$ but with a different classification rule. In the case without any information, we guess the category of $Y$ by choosing $y_j$ with probability $p_{\cdot j}$ for $j=1, ..., b$. This results in a misclassification probability of $1-\sum_{j=1}^{b}p_{\cdot j}^2$. In the case with information about $X$ (say $X=x_i$), we guess $y_j$ with probability $p_{ij}/p_{i\cdot}$ for $j=1, ..., b$. The probability of misclassification given $X=x_i$ is thus $1-\sum_{j=1}^{b}(p_{ij}/p_{i\cdot})^2$. A calculation similar to equation (\ref{eq:LTP}) leads to an unconditional misclassification probability of $1-\sum_{i=1}^{a}\sum_{j=1}^{b}p_{ij}^2/p_{i\cdot}$.
	
	Again, a symmetrized version 
	\begin{align*}
		\tau(X,Y):=\frac{\frac{1}{2}(\sum_{i=1}^{a}\sum_{j=1}^{b}p_{ij}^2/p_{i\cdot}+\sum_{i=1}^{a}\sum_{j=1}^{b}p_{ij}^2/p_{\cdot j}-\sum_{j=1}^{b}p_{\cdot j}^2-\sum_{i=1}^{a}p_{i\cdot}^2)}{1-\frac{1}{2}(\sum_{j=1}^{b}p_{\cdot j}^2+\sum_{i=1}^{a}p_{i\cdot}^2)}
	\end{align*}
	can be easily defined.
	
	Similarly to the contingency coefficients, the improperness of $\lambda$ and $\tau$ is immediate because of the necessity to have two discrete marginal distributions. However, the underlying concept of perfect dependence is still interesting.
	\begin{lemma}[attainability]\label{lem:lambdatau_attainability}
		It holds that $0\le \lambda(X,Y),\tau(X,Y)\le 1$. Additionally, the following statements hold:
		\begin{itemize}
			\item We have $\lambda(X,Y)=0$ if $X$ and $Y$ are independent but not the converse.
			\item We have $\tau(X,Y)=0$ if and only if $X$ and $Y$ are independent.
			\item We have $\lambda(X,Y)=1$ if and only if $\tau(X,Y)=1$ if and only if each row and each column of Table \ref{tab:contingency_table} contains only one nonzero entry.
		\end{itemize}
	\end{lemma}
	Again, we are confronted with an even stricter concept of perfect dependence.
	\begin{lemma}[dependence concepts]\label{lem:lambda_DC}
		If a discrete random vector $(X,Y)$ fulfills $\lambda(X,Y)=1$ (or equivalently $\tau(X,Y)=1$), it also fulfills $MSC(X,Y)=\min\{a,b\}-1$. The reverse is not true.
	\end{lemma}
	
	\subsection{Entropy-Based Measure}
	The last measure to review is the so-called uncertainty coefficient\footnote{In many publications it is mentioned that the coefficient goes back to Henri Theil. However, an original source is hard to find. As a consequence, we refuse to adopt the sometimes used name ``Theil's U''.}. It is originally only defined for two discrete random variables and reads as 
	\begin{align*}
		U(X,Y):=2\left(\frac{H(X)+H(Y)-H(X,Y)}{H(X)+H(Y)}\right),
	\end{align*}
	where \begin{align*}
		H(X):=-\sum_{i=1}^{a}p_{i\cdot}\ln(p_{i\cdot}), \quad H(Y):=-\sum_{j=1}^{b}p_{\cdot j}\ln(p_{\cdot j}), \quad \text{and}\quad  H(X,Y):=-\sum_{i=1}^{a}\sum_{j=1}^{b}p_{ij}\ln(p_{ij}).
	\end{align*}
	We obtain another attainability and dependence concepts lemma.
	\begin{lemma}[attainability]\label{lem:U_attainability}It holds that $0\le U(X,Y)\le 1$. Additionally, the following statements hold:
		\begin{itemize}
			\item We have $U(X,Y)=0$ if and only if $X$ and $Y$ are independent.
			\item We have $U(X,Y)=1$ if and only if each row and each column of Table \ref{tab:contingency_table} contains only one nonzero entry.
		\end{itemize}
	\end{lemma}
	\begin{lemma}[dependence concepts]\label{lem:U_DC}
		If a discrete random vector $(X,Y)$ fulfills $U(X,Y)=1$, it also fulfills $MSC(X,Y)=\min\{a,b\}-1$. The reverse is not true.
	\end{lemma}
    \clearpage
\section{Perfect Dependence via Binary Decomposition}
A self-evident viewpoint towards perfect dependence in contingency tables with nominal random variables is to reduce the problem to the binary case in which nominal random variables can be treated as ordinal. One could, for example, define perfect dependence in such a setting as follows.
\begin{definition}\label{def:perfect_dependence_binary}
    Let $X$ and $Y$ be nominal, that is, $\Omega'=\{\omega'_1, ..., \omega'_a\}$ and $\Omega''=\{\omega''_1, ..., \omega''_b\}$. $X$ and $Y$ are perfectly dependent if every element of $\{\mathds{1}_{\{\omega'_1\}}(X), ..., \mathds{1}_{\{\omega'_a\}}(X)\}\times \{\mathds{1}_{\{\omega''_1\}}(Y), ..., \mathds{1}_{\{\omega''_b\}}(Y)\}$ is perfectly dependent, i.e.\ comonotonic or countermonotonic.
\end{definition}
However, with such a definition, two problems arise. Firstly, we are left without a definition of perfect dependence for the case in which $X$ is nominal and $Y$ is continuous. Secondly, it is easily possible to construct marginal distributions that do not allow for such perfect dependence. 
\begin{example}
    Let $\Omega'=\{A,B,C\}$ and $\Omega''=\{a,b\}$. Consider the marginal probabilities of Table \ref{tab:counterexample_binary}. We start in Table \ref{tab:decomposition1A} with a comonotonic coupling and fill in the corresponding interior fields in Table \ref{tab:counterexample_binary}. In Table \ref{tab:decomposition1B}, we must choose a countermonotonic coupling to respect the marginal probabilities in Table \ref{tab:counterexample_binary}. As a consequence, the decomposition in Table \ref{tab:decomposition1C} can no longer be either co- or countermonotonic. The same logic can be applied in Table \ref{tab:decomposition2} and again we are left with one binary comparison that is not perfectly dependent. Since there are no further possible decompositions, we have proved that our example does not allow for perfect dependence in the sense of Definition \ref{def:perfect_dependence_binary}.
    		\begin{table}[H]
			\centering
     \begin{subtable}{0.49\textwidth}
     \centering
			\begin{tabular}{cccc}
				\toprule
				& a &  b \\
				\midrule 
				$A$& $4/15$&$0$ &$4/15$\\
				$B$&$0$ & $5/15$&$5/15$\\  
				$C$& $3/15$& $3/15$&$6/15$\\
				&$7/15$&$8/15$&1\\
				\bottomrule
			\end{tabular}
			\caption{Decomposition 1}
   \end{subtable}
   \hfill
        \begin{subtable}{0.49\textwidth}
     \centering
			\begin{tabular}{cccc}
				\toprule
				& a &  b \\
				\midrule 
				$A$& $0$& $4/15$&$4/15$\\
				$B$&$5/15$ & $0$&$5/15$\\  
				$C$&$2/15$ & $4/15$&$6/15$\\
				&$7/15$&$8/15$&1\\
				\bottomrule
			\end{tabular}
			\caption{Decomposition 2}
   \end{subtable}
   \caption{Contingency table which does not allow for a decomposition into co- or countermonotonic binary comparisons.}
\label{tab:counterexample_binary}
\end{table}
\begin{table}[H]
  \centering
  \begin{subtable}{0.3\textwidth}
    \centering
			\begin{tabular}{cccc}
				\toprule
				& a &  b \\
				\midrule 
				$A$& $4/15$&$0$ &$4/15$\\
				$A^c$& $3/15$&$8/15$ &$11/15$\\  
				&$7/15$&$8/15$&1\\
				\bottomrule
			\end{tabular}
      \caption{Comonotonic}
      \label{tab:decomposition1A}
\end{subtable}
\hfill
  \begin{subtable}{0.3\textwidth}
    \centering
    \begin{tabular}{cccc}
				\toprule
				& a &  b \\
				\midrule 
				$B$&$0$ & $5/15$&$5/15$\\
				$B^c$&$7/15$ & $3/15$&$10/15$\\  
				&$7/15$&$8/15$&1\\
				\bottomrule
			\end{tabular}
   \caption{Countermonotonic}
\label{tab:decomposition1B}
\end{subtable}
\hfill
  \begin{subtable}{0.3\textwidth}
    \centering
    \begin{tabular}{cccc}
				\toprule
				& a &  b \\
				\midrule 
				$C$&$3/15$ & $3/15$&$6/15$\\
				$C^c$& $4/15$& $5/15$&$9/15$\\  
				&$7/15$&$8/15$&1\\
				\bottomrule
			\end{tabular}
      \caption{Neither co- nor countermon.}
      \label{tab:decomposition1C}
\end{subtable}
\caption{Decomposition 1}
\label{tab:decomposition1}
\end{table}
\begin{table}[H]
  \centering
  \begin{subtable}{0.3\textwidth}
    \centering
			\begin{tabular}{cccc}
				\toprule
				& a &  b \\
				\midrule 
				$A$& $0$& $4/15$&$4/15$\\
				$A^c$& $7/15$& $4/15$&$11/15$\\  
				&$7/15$&$8/15$&1\\
				\bottomrule
			\end{tabular}
      \caption{Countermonotonic}
      \label{tab:decomposition2A}
\end{subtable}
\hfill
  \begin{subtable}{0.3\textwidth}
    \centering
    \begin{tabular}{cccc}
				\toprule
				& a &  b \\
				\midrule 
				$B$& $5/15$& $0$&$5/15$\\
				$B^c$& $2/15$& $8/15$&$10/15$\\  
				&$7/15$&$8/15$&1\\
				\bottomrule
			\end{tabular}
   \caption{Comonotonic}
\label{tab:decomposition2B}
\end{subtable}
\hfill
  \begin{subtable}{0.3\textwidth}
    \centering
    \begin{tabular}{cccc}
				\toprule
				& a &  b \\
				\midrule 
				$C$& $2/15$& $4/15$&$6/15$\\
				$C^c$&$5/15$ & $4/15$&$9/15$\\  
				&$7/15$&$8/15$&1\\
				\bottomrule
			\end{tabular}
      \caption{Neither co- nor countermon.}
      \label{tab:decomposition2C}
\end{subtable}
\caption{Decomposition 2}
\label{tab:decomposition2}
\end{table}
\end{example}
\clearpage
\section{Proofs}
This section provides the proofs for all the conceptual results of Sections \ref{sec:perfect_dependence}, \ref{sec:improper_measures} and \ref{sec:proper_measures} as well as Appendix \ref{sec:AppendixImproperMeasures} in Subsection \ref{subsec-App:conceptual}. After that, it presents the proofs for all the asymptotic results of Section \ref{sec:statistical_inference} in Subsection \ref{subsec-App:asymptotic}.
\subsection{Proofs for the Results of Sections \ref{sec:perfect_dependence}, \ref{sec:improper_measures}, \ref{sec:proper_measures} and Appendix \ref{sec:AppendixImproperMeasures}}\label{subsec-App:conceptual}
	\begin{proof}[Proof of Lemma \ref{lem:perfectdependence}]
		Case 1: ``$(ii)\Rightarrow (i)$''. By assumption it holds that $F_{X^{s_a},Y}(x,y)=\min\left\{F_{X^{s_a}}(x),F_Y(y)\right\}$. This is equivalent to $(X^{s_a},Y)\sim (F^{\leftarrow}_{X^{s_a}}(U), F^{\leftarrow}_Y(U))$, where $U\sim U(0,1)$ and $F^{\leftarrow}(u):=\inf\{x:F(x)\ge u\}$ for some CDF $F$ \citep{Embrechts2002}. Denote now by $s_a^-\in S_a$ the complementary permutation of $s_a$, i.e.\ $$s_a^-=(a-s_a(1)+1, ..., a-s_a(a)+1).$$
		It follows that $\left(X^{s_a^-},Y\right)\sim \left(F^{\leftarrow}_{X^{s_a^-}}(1-U), F^{\leftarrow}_Y(U)\right)$, which is equivalent to $F_{X^{s_a^-},Y}(x,y)=\max\left\{F_{X^{s_a}}(x)+F_Y(y)-1,0\right\}$. The proof for the opposite direction works similarly.
		
		Case 2: Perform the same procedure as in the proof for case 1 while holding $s_b$ constant.
	\end{proof}
    \begin{proof}[Proof of Lemma \ref{lem:MSC_attainability}]
    The first statement is immediate from the definitions of independence and $MSC$.
    
    For the second statement, assume w.l.o.g.\ that $a\ge b$. We can rewrite
    \begin{align*}
        MSC(X,Y)=&\ \sum_{i=1}^{a}\sum_{j=1}^{b}\frac{(\mathbb{P}(X=x_i, Y=y_j)-\mathbb{P}(X=x_i)\mathbb{P}(Y=y_j))^2}{\mathbb{P}(X=x_i)\mathbb{P}(Y=y_j)}\\
        =&\ \sum_{i=1}^{a}\sum_{j=1}^{b}\frac{\mathbb{P}(X=x_i, Y=y_j)^2}{\mathbb{P}(X=x_i)\mathbb{P}(Y=y_j)}-2\sum_{i=1}^{a}\sum_{j=1}^{b}\mathbb{P}(X=x_i, Y=y_j)\\
        &+\sum_{i=1}^{a}\sum_{j=1}^{b}\mathbb{P}(X=x_i)\mathbb{P}(Y=y_j).
    \end{align*}
    The latter two sums are 1 and therefore it follows that 
    \begin{align*}
        MSC(X,Y)
        =&\ \sum_{i=1}^{a}\sum_{j=1}^{b}\frac{\mathbb{P}(X=x_i, Y=y_j)^2}{\mathbb{P}(X=x_i)\mathbb{P}(Y=y_j)}-1.
    \end{align*}
    We can further rewrite 
    \begin{align*}
        \sum_{i=1}^{a}\sum_{j=1}^{b}\frac{\mathbb{P}(X=x_i, Y=y_j)^2}{\mathbb{P}(X=x_i)\mathbb{P}(Y=y_j)}
        = \sum_{j=1}^{b}\frac{1}{\mathbb{P}(Y=y_j)}\sum_{i=1}^{a}\underbrace{\frac{\mathbb{P}(X=x_i, Y=y_j)}{\mathbb{P}(X=x_i)}}_{\le 1}\mathbb{P}(X=x_i, Y=y_j)
        \le b. 
    \end{align*}
    It is obvious that this upper bound is reached if and only if $\mathbb{P}(X=x_i, Y=y_j)=\mathbb{P}(X=x_i)$ for all $i=1, ..., a$ and $j=1, ..., b$ for which $\mathbb{P}(X=x_i, Y=y_j)\ne 0$ holds. Since there can always be only one $j$ for which this condition holds, the other joint probabilities have to be 0.
    \end{proof}
    \begin{proof}[Proof of Lemma \ref{lem:V_attainability}]
    We only prove the third statement. W.l.o.g.\ we assume $a\ge b$. It holds that
    \begin{align*}
        \min\{a-1,b-1\}=b-1=\sqrt{(b-1)^2}\le \sqrt{(b-1)(a-1)}
    \end{align*}
    with equality if and only if $a=b$.
	\end{proof}
    \begin{proof}[Proof of Lemma \ref{lem:MSC_DC}]
    We know that $MSC(X,Y)=\min\{a,b\}-1$ if and only if each row (if $a\ge b$) or each column (if $a\le b$) contains only one non-zero entry. If we have two nominal variables, we can thus choose a numbering $(s_a, s_b)$ leading to $(X^{s_a},Y^{s_b})=(F^{\leftarrow}_{X^{s_a}}(U), F^{\leftarrow}_{Y^{s_b}}(U))$. The same holds true if we only have one variable of a nominal scale. 

    For the second statement, observe the following two counterexamples:
    \begin{enumerate}
        \item Consider $a=b=3$ and $\mathbb{P}(X=x_1)=0.1,\ \mathbb{P}(X=x_2)=0.7, \ \mathbb{P}(X=x_3)=0.2,\ \mathbb{P}(Y=y_1)=0.3,\ \mathbb{P}(Y=y_2)=0.6$ and $\mathbb{P}(Y=y_3)=0.1$. With such marginal probabilities, it is impossible to form a contingency table in which every row and every column contain only one non-zero entry.
        \item Table \ref{tab:MSC_DC_counterexamplea} and \ref{tab:MSC_DC_counterexampleb} constitute a counterexample.
    \end{enumerate}
    \begin{table}[tb]
    \centering
    \begin{subtable}{0.49\textwidth}
        \centering
        \small
        \begin{tabular}{ccccc}
            \toprule
            & $y_1$ &  $y_2$ & $y_3$ \\
            \midrule 
            $x_1$&$0.4$&$0.2$&$0.1$&$0.7$\\
            $x_2$&$0.2$&$0$&$0$&$0.2$\\  
            $x_3$&$0.1$&$0$&$0$&$0.1$\\
            &$0.7$&$0.2$&$0.1$&1\\
            \bottomrule
        \end{tabular}
        \caption{$MSC(X,Y)\ne \min\{a,b\}-1$ and perfect dependence as in Definition \ref{def:perfect_dependence}.}
        \label{tab:MSC_DC_counterexamplea}
        \end{subtable}
        \hfill
        \begin{subtable}{0.49\textwidth}
        \centering
        \small
        \begin{tabular}{ccccc}
            \toprule
            & $y_1$ &  $y_2$ & $y_3$ \\
            \midrule 
            $x_1$&$0.7$&$0$&$0$&$0.7$\\
            $x_2$&$0$&$0.2$&$0$&$0.2$\\  
            $x_3$&$0$&$0$&$0.1$&$0.1$\\
            &$0.7$&$0.2$&$0.1$&1\\
            \bottomrule
        \end{tabular}
        \caption{$MSC(X,Y)=\min\{a,b\}-1$ and perfect dependence as in Definition \ref{def:perfect_dependence}.}
        \label{tab:MSC_DC_counterexampleb}
        \end{subtable}
        \caption{Fictitious counterexample supporting the proof for Lemma \ref{lem:MSC_DC}.}
        \label{tab:MSC_DC_counterexample}
    \end{table}
	\end{proof}
    \begin{proof}[Proof of Proposition \ref{prop:propercase1}]
    Axioms (A) and (F) are clear.

    For axiom (B), note that the complementary permutation (see the proof for Lemma \ref{lem:perfectdependence} for a definition) constitutes a strictly monotonic transformation and it holds that $\gamma(X^{s_a}, Y)=-\gamma(X^{s_a^-},Y)$. The weak positivity of $\gamma^*$ follows.

    For axiom (C), suppose that $X$ and $Y$ are independent. Then, $X^{s_a}$ and $Y$ are also independent for every $s_a \in S_a$. Therefore, $\gamma^*(X,Y)=0$ follows. For the other direction, consider the following counterexample. Let $Y\sim \mathcal{N}(0,1)$ and $X$ being defined via $\p(X=A\mid 3\le |Y|)=1,\ \p(X=B\mid 2\le |Y|<3)=1$ and $\p(X=C\mid 2> |Y|)=1$. Clearly, $X$ and $Y$ are not independent but it still holds that $\gamma^*(X,Y)=0$.

    Axiom (D) follows because if $X$ and $Y$ are perfectly dependent, the maximum operator will find $s_a \in S_a$ such that $\gamma(X^{s_a}, Y)=1$. This is true because $\gamma$ is 1 if and only if the supplied bivariate CDF is an upper Fréchet-Hoeffding bound (see Remark 2.8 in \cite{mesfioui2005properties}). Conversely, if $\gamma^{*}(X,Y)=1$, there exists $s_a \in S_a$ such that $\gamma(X^{s_a}, Y)=1$ and therefore, $(X^{s_a}, Y)$ are comonotonic which is equivalent to $X$ and $Y$ being perfectly dependent.

    Axiom $(E)$ follows for strictly increasing transformations because of the rank-based nature of $Y$. For strictly decreasing transformations $g$, we can utilize that $\gamma(X^{s_a}, Y)=-\gamma(X^{s_a^-}, Y)=\gamma(X^{s_a^-}, g(Y))$, where $s_a^-$ is again the complementary permutation from Lemma \ref{lem:perfectdependence}.
    \end{proof}
    \begin{proof}[Proof of Proposition \ref{prop:propercase2}]
    Axioms (A) and (F) are clear. 

    For axiom (B), note that the complementary permutation (see the proof for Lemma \ref{lem:perfectdependence} for a definition) constitutes a strictly monotonic transformation and it holds that $\gamma(X^{s_a}, Y^{s_b})=-\gamma(X^{s_a^-}, Y^{s_b})=-\gamma(X^{s_a}, Y^{s_b^-})=\gamma(X^{s_a^-}, Y^{s_b^-})$. The weak positivity of $\gamma^*$ follows.
    
    Axiom (D) follows because if $X$ and $Y$ are perfectly dependent, the maximum operator will find $s_a \in S_a$ and $s_b \in S_b$ such that $\gamma(X^{s_a}, Y^{s_b})=1$ (see again Remark 2.8 in \cite{mesfioui2005properties}). Conversely, if $\gamma^{*}(X,Y)=1$, there exist $s_a \in S_a$ and $s_b \in S_b$ such that $\gamma(X^{s_a}, Y^{s_b})=1$ and therefore, $(X^{s_a}, Y^{s_b})$ is comonotonic or countermonotonic which is equivalent to $X$ and $Y$ being perfectly dependent.

    For axiom (C), consider the following reasoning. Suppose $X$ and $Y$ are independent. Then, $X^{s_a}$ and $Y^{s_b}$ are also independent for every $s_a \in S_a$ and $s_b \in S_b$. Therefore, $\gamma^*(X,Y)=0$ follows. 

    Conversely, assume $\gamma^*(X,Y)=0$. Then, the reasoning in the proof for axiom (B) applies and we have $\gamma(X^{s_a}, Y^{s_b})=0$ for all $s_a\in S_a$ and $s_b\in S_b$. We now proceed in three steps. Firstly, we show the binary case $(a=b=2)$. After that, we prove the result for arbitrary $a\in \mathbb{N}_{\ge 2}$ and $b=2$ and then for arbitrary $a\in \mathbb{N}_{\ge 2}$ and $b\in \mathbb{N}_{\ge 2}$. For the binary case, simple calculations yield $\tau^{K}\left(\mathds{1}_{\omega_1'}, \mathds{1}_{\omega_1''}\right)=2(\mathbb{P}(X=\omega_1', Y=\omega_1'')-\mathbb{P}(X=\omega_1') \mathbb{P}(Y=\omega_1''))$, where $\tau^{K}$ denotes Kendall's tau. Therefore, by assumption we get $$\mathbb{P}(X=\omega_1', Y=\omega_1'')=\mathbb{P}(X=\omega_1') \mathbb{P}(Y=\omega_1''), \mathbb{P}(X=\omega_1', Y=\omega_2'')=\mathbb{P}(X=\omega_1') \mathbb{P}(Y=\omega_2''),$$ $$\mathbb{P}(X=\omega_2', Y=\omega_1'')=\mathbb{P}(X=\omega_2') \mathbb{P}(Y=\omega_1'') \text{ and } \mathbb{P}(X=\omega_2', Y=\omega_2'')=\mathbb{P}(X=\omega_2') \mathbb{P}(Y=\omega_2'').$$ This is equivalent to $X$ and $Y$ being independent. 

    For arbitrary $a\in \mathbb{N}_{\ge 2}$ and $b=2$, our assumption implies for fixed $s_a$ and $s_b$ that \begin{align*}&\quad \ 2\sum_{j=1}^{a-1}\mathbb{P}(X^{s_a}=j,Y^{s_b}=1)\sum_{i=j+1}^{a}\mathbb{P}(X^{s_a}=i,Y^{s_b}=2)\\&=\mathbb{P}((X^{s_a}-\widetilde{X}^{s_a})(Y^{s_b}-\widetilde{Y}^{s_b})>0)\\ &=\mathbb{P}((X^{s_a}-\widetilde{X}^{s_a})(Y^{s_b}-\widetilde{Y}^{s_b})<0)\\ &=2\sum_{j=1}^{a-1}\mathbb{P}(X^{s_a}=j,Y^{s_b}=2)\sum_{i=j+1}^{a}\mathbb{P}(X^{s_a}=i,Y^{s_b}=1),
    \end{align*}
    where $(\widetilde{X}^{s_a}, \widetilde{Y}^{s_b})$ denotes an independent copy of $(X^{s_a}, Y^{s_b})$.
    We can rewrite \begin{align}\nonumber
    &\mathbb{P}(X^{s_a}=1,Y^{s_b}=1)\mathbb{P}(X^{s_a}>1,Y^{s_b}=2)+\sum_{j=2}^{a-1}\mathbb{P}(X^{s_a}=j,Y^{s_b}=1)\sum_{i=j+1}^{a}\mathbb{P}(X^{s_a}=i,Y^{s_b}=2)\\
    &=\mathbb{P}(X^{s_a}=1,Y^{s_b}=2)\mathbb{P}(X^{s_a}>1,Y^{s_b}=1)+\sum_{j=2}^{a-1}\mathbb{P}(X^{s_a}=j,Y^{s_b}=2)\sum_{i=j+1}^{a}\mathbb{P}(X^{s_a}=i,Y^{s_b}=1).\label{eq:orig_perm}
    \end{align}
    The double sum at the left hand side of equation (\ref{eq:orig_perm}) can be reordered as 
    \begin{align*}
    \sum_{j=2}^{a-1}\sum_{i=j+1}^{a}\mathbb{P}(X^{s_a}=j,Y^{s_b}=1)\mathbb{P}(X^{s_a}=i,Y^{s_b}=2)=\sum_{i=3}^{a}\sum_{j=2}^{i-1}\mathbb{P}(X^{s_a}=i,Y^{s_b}=2)\mathbb{P}(X^{s_a}=j,Y^{s_b}=1)
    \end{align*}
    
    By assumption, the equation still holds if we permute the $a-1$ largest values for $X^{s_a}$ according to the complementary permutation $(a\mapsto 2, a-1 \mapsto 3, ..., 2\mapsto a)$. Therefore, it follows that
    \begin{align*}
    \sum_{i=3}^{a}\sum_{j=2}^{i-1}\mathbb{P}(X^{s_a}=i,Y^{s_b}=2)\mathbb{P}(X^{s_a}=j,Y^{s_b}=1)=\sum_{j=2}^{a-1}\sum_{i=j+1}^{a}\mathbb{P}(X^{s_a}=j,Y^{s_b}=2)\mathbb{P}(X^{s_a}=i,Y^{s_b}=1)
    \end{align*}
    
    We arrive at a second equation
    \begin{align}\nonumber
    &\mathbb{P}(X^{s_a}=1,Y^{s_b}=1)\mathbb{P}(X^{s_a}>1,Y^{s_b}=2)+\sum_{j=2}^{a-1}\mathbb{P}(X^{s_a}=j,Y^{s_b}=2)\sum_{i=j+1}^{a}\mathbb{P}(X^{s_a}=i,Y^{s_b}=1)\\
    &=\mathbb{P}(X^{s_a}=1,Y^{s_b}=2)\mathbb{P}(X^{s_a}>1,Y^{s_b}=1)+\sum_{j=2}^{a-1}\mathbb{P}(X^{s_a}=j,Y^{s_b}=1)\sum_{i=j+1}^{a}\mathbb{P}(X^{s_a}=i,Y^{s_b}=2),\label{eq:cont_perm}
    \end{align}
    because the double sum at the right hand side of equation (\ref{eq:orig_perm}) gets transformed in the opposite direction. Subtracting both equations from each other, it follows that
    \begin{align*}
    &\mathbb{P}(X^{s_a}=1,Y^{s_b}=1)\mathbb{P}(X^{s_a}>1,Y^{s_b}=2)-\mathbb{P}(X^{s_a}=1,Y^{s_b}=2)\mathbb{P}(X^{s_a}>1,Y^{s_b}=1)\\
    &=\mathbb{P}(X^{s_a}=1,Y^{s_b}=2)\mathbb{P}(X^{s_a}>1,Y^{s_b}=1)-\mathbb{P}(X^{s_a}=1,Y^{s_b}=1)\mathbb{P}(X^{s_a}>1,Y^{s_b}=2),
    \end{align*}
    which is equivalent to 
    \begin{align*}
    \mathbb{P}(X^{s_a}=1,Y^{s_b}=2)\mathbb{P}(X^{s_a}>1,Y^{s_b}=1)=\mathbb{P}(X^{s_a}=1,Y^{s_b}=1)\mathbb{P}(X^{s_a}>1,Y^{s_b}=2).
    \end{align*}
    This in turn implies uncorrelatedness of the events $\{X^{s_a}=1\}$ and $\{Y^{s_b}=2\}$ and their corresponding complements. In the binary case, uncorrelatedness is equivalent to independence. Since $s_a$ was arbitrary, we could start with any other permutation and arrive at the result that any $\mathds{1}_{\{\omega_1'\}}, ..., \mathds{1}_{\{\omega_a'\}}$ is independent of $Y$. As a result, we get $\mathbb{P}(A\cap B)=\mathbb{P}(A)\mathbb{P}(B)$ for any $A\in \mathcal{P}(\Omega')$ and $B\in \mathcal{P}(\Omega'')$, which is equivalent to independence between $X$ and $Y$.

    For the general case, we can apply the previous reasoning $b$ times and arrive at the conclusion that \begin{align*}
    \mathbb{P}(X=\omega'_i, Y=\omega''_j)=\mathbb{P}(X=\omega'_i)\mathbb{P}(Y=\omega''_j)\quad \forall (i,j)\in \{1, ..., a\}\times \{1, ..., b\}
    \end{align*}
    This is enough for independence between $X$ and $Y$.
	\end{proof}
    \begin{proof}[Proof of Lemma \ref{lem:maximalcorrelation}]
    The set of all bijective Borel-measurable functions has infinite cardinality. Since any two functions $f$ and $f'$ who order the values of $X$ in the same way also lead to the same coefficient value, we only have to consider either $f$ or $f'$ and not both. Therefore, a function class with one representative for each numbering suffices for the maximization problem to remain unchanged. Then, we maximize over a set of finite cardinality and the supremum can be replaced by a maximum.
	\end{proof}
    \begin{proof}[Proof of Lemma \ref{lem:lambdatau_attainability}]
    We start with statement 1. Under independence it holds that
    \begin{align*}\sum_{i=1}^{a}p_{im}+\sum_{j=1}^{b}p_{mj}-p_{\cdot m}-p_{m \cdot}=p_{\cdot m}\sum_{i=1}^{a}p_{i\cdot}+p_{m\cdot}\sum_{j=1}^{b}p_{\cdot j}-p_{\cdot m}-p_{m \cdot}=0.
    \end{align*}
    \begin{table}[tb]
        \centering
        \small
        \begin{tabular}{ccccc}
            \toprule
            & 1 &  2 & 3 \\
            \midrule $A$&$0.3$&$0.12$&$0.08$&$0.5$\\
            $B$&$0.12$&$0.11$&$0.07$&$0.3$\\  
            $C$&$0.08$&$0.07$&$0.05$&$0.2$\\
            &$0.5$&$0.3$&$0.2$&1\\
            \bottomrule
        \end{tabular}
        \caption{Fictitious counterexample supporting the proof for Lemma \ref{lem:lambdatau_attainability}.}
        \label{tab:lambdatau_att_counterexample}
    \end{table}
    As a counterexample that $\lambda=0$ does not imply independence, consider Table \ref{tab:lambdatau_att_counterexample}. It holds that 
    \begin{align*}
        \sum_{i=1}^{a}p_{im}+\sum_{j=1}^{b}p_{mj}-p_{\cdot m}-p_{m \cdot}=2(0.3+0.12+0.08-0.5)=0,
    \end{align*}
    but the involved random variables are clearly not independent.
    
    For the second statement, note that $\tau$ can be rewritten as 
    \begin{align*}
        \tau(X,Y)=\frac{\tau_y(1-\sum_{j=1}^{b}p_{\cdot j}^2)+\tau_x(1-\sum_{i=1}^{a}p_{i\cdot}^2)}{2-\sum_{j=1}^{b}p_{\cdot j}^2-\sum_{i=1}^{a}p_{i\cdot}^2}.
    \end{align*}
    Therefore, $\tau=0$ is equivalent to $\tau_y=0$ and $\tau_x=0$.
    We reformulate the numerator in $\tau_y$ as
    \begin{align*}
        \sum_{i=1}^{a}\sum_{j=1}^{b}\frac{p_{ij}^2}{p_{i\cdot}}-\sum_{j=1}^{b}p_{\cdot j}^2&=\sum_{i=1}^{a}\sum_{j=1}^{b}\frac{p_{ij}^2}{p_{i\cdot}}-\sum_{i=1}^{a}\sum_{j=1}^{b}p_{i\cdot}p_{\cdot j}^2\\
        &=\sum_{i=1}^{a}\sum_{j=1}^{b}\frac{p_{ij}^2}{p_{i\cdot}}-2p_{i\cdot}p_{\cdot j}^2+p_{i\cdot}p_{\cdot j}^2\\
        &=\sum_{i=1}^{a}\sum_{j=1}^{b}\frac{p_{ij}^2}{p_{i\cdot}}-2p_{ij}p_{\cdot j}+p_{i\cdot}p_{\cdot j}^2\\
        &=\sum_{i=1}^{a}\sum_{j=1}^{b}\frac{(p_{ij}-p_{i\cdot}p_{\cdot j})^2}{p_{i\cdot}}.
    \end{align*}
    Note that the third equality did not use the independence between $X$ and $Y$ since it holds that $$\sum_{i=1}^{a}\sum_{j=1}^{b}p_{i\cdot}p_{\cdot j}^2=\sum_{j=1}^{b}p_{\cdot j}^2\sum_{i=1}^{a}p_{i\cdot}=\sum_{j=1}^{b}p_{\cdot j}^2=\sum_{j=1}^{b}p_{\cdot j}^2\sum_{i=1}^{a}\frac{p_{ij}}{p_{\cdot j}}=\sum_{i=1}^{a}\sum_{j=1}^{b}p_{ij}p_{\cdot j}$$
    without such an assumption. The same reformulation can be conducted for $\tau_x$ and the result follows.
    
    Statement 3 follows by the following reasoning: We have \begin{align*}
        \lambda(X,Y)=1\Leftrightarrow \sum_{i=1}^{a}p_{im}+\sum_{j=1}^{b}p_{mj}=2
    \end{align*}
    and the latter equality is equivalent to the condition in the statement. Similarly, it holds that
    \begin{align*}
        \tau(X,Y)=1 \Leftrightarrow \sum_{i=1}^{a}\sum_{j=1}^{b}p_{ij}^2/p_{i\cdot}+\sum_{i=1}^{a}\sum_{j=1}^{b}p_{ij}^2/p_{\cdot j}=2.
    \end{align*}
    Since $p_{ij}/p_{i\cdot}\le 1$, the inequality $\sum_{i=1}^{a}\sum_{j=1}^{b}p_{ij}^2/p_{i\cdot}\le\sum_{i=1}^{a}\sum_{j=1}^{b}p_{ij}=1$ holds where equality holds if and only if $p_{ij}=p_{i\cdot} \ \forall \ i=1, ..., a$ and $j=1, ..., b$ for which $p_{ij}\ne 0$. A similar argument holds for the second summand and the statement follows.
	\end{proof}
    \begin{proof}[Proof of Lemma \ref{lem:lambda_DC}]
    The lemma follows by invoking Lemma \ref{lem:MSC_attainability} and Lemma \ref{lem:lambdatau_attainability}.
	\end{proof}
    \begin{proof}[Proof of Lemma \ref{lem:U_attainability}]
    For both statements, see equations (14.4.11)-(14.4.16) and (14.4.19) in \citet{press1992C}.
	\end{proof}
    \begin{proof}[Proof of Lemma \ref{lem:U_DC}]
    The proof is immediate by Lemma \ref{lem:lambda_DC} and Lemma \ref{lem:U_attainability}.
	\end{proof}
\subsection{Proofs for the Results of Section \ref{sec:statistical_inference}}\label{subsec-App:asymptotic}
	\begin{proof}[Proof of Proposition \ref{prop:gamma*_consistency}]
    We only prove the proposition for the case of two nominal random variables. The proof for case 1 is analogous. With probability 1, there exists a sample size $N\in \mathbb{N}$ for which every element of $\Omega'$ and $\Omega''$ occurs in the sample. Hence, for $n\ge N$ we have $k=a$ and $l=b$ in the notation of Definition \ref{def:estim}. Since $\widehat{\gamma}_n$ is a consistent estimator for $\gamma$ \citep{pohle2024inference}, $\widehat{\gamma}_n(X^{s_a},Y^{s_b})$ is a consistent estimator for $\gamma(X^{s_a},Y^{s_b})$ for fixed permutations $s_a\in S_a$ and $s_b \in S_b$. If we write all coefficient values and estimates over which subsequently the maximum is taken in vectors, it holds that
    \begin{align}\label{eq:estim}
        \begin{pmatrix}
            \vdots\\
            \widehat{\gamma}_n(X^{s_a},Y^{s_b})\\
            \vdots
        \end{pmatrix}
        \overset{p}{\rightarrow}
        \begin{pmatrix}
            \vdots\\
            \gamma(X^{s_a},Y^{s_b})\\
            \vdots
        \end{pmatrix}
    \end{align}
    because consistency of a vector is equivalent to element-wise consistency \citep{lehmann1999elements}. The maximum function is continuous and thus, the continuous mapping theorem delivers the claim.
	\end{proof}
    \begin{proof}[Proof of Proposition \ref{prop:jointasymptotics}]
    Assume w.l.o.g.\ that we are in case 2. We use the multivariate CLT for U-statistics given in the appendix of \citet{pohle2024inference}. Our underlying random vector consists of all variables $X^{s_a}$ and $Y^{s_b}$ with $s_a\in S_a$ and $s_b\in S_b$ and is thus of dimension $a!+ b!$. We denote Kendall's $\tau$ by $\tau^K:=\tau^K(X^{s_a}, Y^{s_b})=\mathbb{P}((X^{s_a}-\widetilde{X}^{s_a})(Y^{s_b}-\widetilde{Y}^{s_b})>0)-\mathbb{P}((X^{s_a}-\widetilde{X}^{s_a})(Y^{s_b}-\widetilde{Y}^{s_b})<0)$ and $\nu:=\nu(X^{s_a}, Y^{s_b})=\mathbb{P}((X^{s_a}-\widetilde{X}^{s_a})(Y^{s_b}-\widetilde{Y}^{s_b})=0)$, where $(\widetilde{X}^{s_a},\widetilde{Y}^{s_b})$ is an independent copy of $(X^{s_a},Y^{s_b})$. The vector of U-statistics thus consists of $a!\cdot b!$ pairs $(\widehat{\tau}_n, \widehat{\nu}_n)$ and is thus of length $2\cdot a!\cdot b!$. In order to be able to distinguish those pairs from each other, we choose an arbitrary numbering of all elements of $S_a\times S_b$ and write $(\hdots \ \widehat{\tau}_n(i)\ \widehat{\nu}_n(i)\ \hdots)^{\top}$, where $1\le i \le a!\cdot b!$. It follows that
    \begin{align*}\sqrt{n}\left(
        \begin{pmatrix}
            \vdots\\
            \widehat{\tau}_n(i)\\
            \widehat{\nu}_n(i)\\
            \vdots
        \end{pmatrix}-
        \begin{pmatrix}
            \vdots\\
            \tau(i)\\
            \nu(i)\\
            \vdots
        \end{pmatrix}
        \right)\overset{d}{\rightarrow} \mathcal{N}_{2\cdot a!\cdot b!}(0,\Sigma_U^2) \quad \text{with} \quad \Sigma_U^2=(4\sigma_{v(i) w(j)})_{1\le i,j\le a!\cdot b!},
    \end{align*}
    where 
    \begin{align*}
        \sigma_{v(i) w(j)}=\E[k_1^{(v(i))}(X^{s_a(i)}, Y^{s_b(i)})k_1^{(w(j))}(X^{s_a(j)}, Y^{s_b(j)})],\quad v,w\in \{\tau, \nu\}
    \end{align*}
    and 
    \begin{align*}
        k_1^{(\tau(i))}(x, y)&=4G_{X^{s_a(i)}, Y^{s_b(i)}}(x,y)-2(G_{X^{s_a(i)}}(x)+G_{Y^{s_b(i)}}(y))+1-\tau(i),\\
        k_1^{(\nu(i))}(x, y)&=\p(X^{s_a(i)}=x)+\p(Y^{s_b(i)}=y)-\p(X^{s_a(i)}=x, Y^{s_b(i)}=y)-\nu(i).
    \end{align*}
    We now want to apply the Delta Method \citep[Theorem 3.1]{vandervaart2000} to obtain our result. For any fixed numbering $i$, it holds that $\gamma(i):=\gamma(X^{s_a(i)},Y^{s_b(i)})=f(\tau(i), \nu(i))=\frac{\tau(i)}{1-\nu(i)}$. Also, we obtain the derivatives \begin{align*}
        \frac{d\gamma(i)}{d\tau(i)}=\frac{1}{1-\nu(i)}, \quad \frac{d\gamma(i)}{d\nu(i)}=\frac{\tau(i)}{(1-\nu(i))^2}, \quad \frac{d\gamma(i)}{d\tau(j)}=0\quad \text{and} \quad \frac{d\gamma(i)}{d\nu(j)}=0
    \end{align*}
    for $j\ne i$. If we collect those derivatives in a matrix of dimension $a!\cdot b!\times 2\cdot a!\cdot b!$, we obtain
    \begin{align*}
        A^2:=\begin{pmatrix}
            \frac{1}{1-\nu(1)}&\frac{\tau(1)}{(1-\nu(1))^2}&0&0&0&0&\hdots&0\\
            0&0&\frac{1}{1-\nu(2)}&\frac{\tau(2)}{(1-\nu(2))^2} &0&0&\hdots&0\\
            \vdots&\vdots&\vdots&\vdots&\ddots&\ddots&\vdots&\vdots\\
            0&0&0&0&0&\hdots&\frac{1}{1-\nu(a!\cdot b!)}&\frac{\tau(a!\cdot b!)}{(1-\nu(a!\cdot b!))^2}
        \end{pmatrix}
    \end{align*}
    The statement in the proposition follows with $\Sigma^2=A\Sigma_U^2A^{\top}$ which holds because of the Delta Method.
	\end{proof}
    \begin{proof}[Proof of Proposition \ref{prop:asymptotics}]
    We start with the first statement. There exists some $N\in \mathbb{N}$ for which it holds that for $n\ge N$, $\widehat{\gamma}^*_n$ will almost surely always choose the population numbering $s_a^{(*)}$ $(s_b^{(*)})$. Therefore, $\sqrt{n} \left( \widehat{\gamma}^*_n -\gamma^* \right)$ is asymptotically equivalent to $\sqrt{n} \left( \widehat{\gamma}_n^{(*)} -\gamma^{(*)} \right)$ and the statement follows.

    For the second statement, note that $\widehat{\gamma}^*_n$ will also asymptotically oscillate between the $d$ population numberings. However, we almost surely find a $N\in \mathbb{N}$ such that for any sample size $n\ge N$, $\widehat{\gamma}^*_n$ almost surely always chooses one of the $d$ population numberings. The statement follows by the continuous mapping theorem because the maximum function is continuous.
	\end{proof}
 \begin{proof}[Proof of Lemma \ref{lem:variance_estimator_consistency}]
 Follows directly from Assumption \ref{ass:simplifying}, Glivenko-Cantelli's Theorem and the proof that $\widehat{\sigma}^2_{\gamma}$ is consistent for $\sigma^2_{\gamma}$ in the case of real-valued random variables \citep{pohle2024inference}.
 \end{proof}
 \begin{proof}[Proof of Lemma \ref{lem:globalFTest}]
    If $b_j=0\ \forall j\in\{2, ..., a\}$, then $Y=b_1+U\ a.s.$ and the independence follows by assumption. For the other direction, assume by contraposition that $b_j\ne 0$ for some $j\in \{2, ..., a\}$. Then, there exists some $A\in \mathcal{B}(\mathbb{R})$ for which $\mathbb{P}(Y\in A|X=\omega_j')\ne \mathbb{P}(Y\in A)$ holds.
\end{proof}
\begin{proof}[Proof of Lemma \ref{lem:gamma*asymptotics_independence}]
    By the Propositions \ref{prop:propercase1} and  \ref{prop:propercase2}, any population numbering will deliver $\gamma=0$. Therefore, we have to include every one-dimensional marginal distribution computed in Proposition \ref{prop:jointasymptotics} in our limiting distribution.
\end{proof}
\clearpage
\section{Computational Implementation}\label{sec-App:Computation}
The estimators introduced in Definition \ref{def:estim} involve a maximization and therefore, their computation is non-trivial. A brute-force approach would calculate $\widehat{\gamma}$ for each of the $k!$ (case 1) or $k!l!$ (case 2) bivariate samples and then choose the maximum value. However, due to the quickly increasing nature of the factorial, this algorithm gets infeasible for larger values of $k$ (and $l$). 

One key observation for accelerating the computation time in both cases is that for a given sample, maximizing $\widehat{\gamma}$ is equivalent to maximizing the number of concordant pairs. That is because the number of joint ties (and thus the sum of the number of concordant and discordant pairs) that determine the denominator of $\widehat{\gamma}$ is fixed irrespective of which permutation we assign.\footnote{Note that altering the dependence structure in a bivariate (empirical) distribution while at the same time keeping the marginal tie probabilities constant (which is what reassigning permutations is effectively doing) does in general not lead to identical joint tie probabilities (see Appendix C in \cite{pohle2024inference}). In this specific situation however, the invariance holds.}

For the case with only one nominal variable (case 1), we can circumvent the $\mathcal{O}(k!)$ complexity by employing a dynamic programming approach that yields the result in $\mathcal{O}(k2^k)$ time. We proceed as follows:
\begin{enumerate}
    \item Compute a matrix indexed by the nominal values of $X$: $$H[x_l,x_m]:=\#\{(i,j):Y_i<Y_j,X_i=x_l, X_j=x_m\}$$ with $l,m\in \{1, ..., k\}$ and $l\ne m$. Its $(x_l,x_m)$-entry amounts to the number of concordant pairs that would occur had $x_l$ been assigned a smaller number than $x_m$. We set $H[x_l,x_m]=0$ for $l,m\in \{1, ..., k\}$ and $l=m$.
    \item Compute the total number of non-tied pairs: $T=\sum_{l=1}^k\sum_{m=1}^kH[x_l,x_m]$
    \item For a set $S\in \mathcal{P}(\{x_1, ..., x_k\})$, where $\mathcal{P}$ denotes the power set operator, define \begin{align*}dp[S]:=\text{The most concordant pairs one can get using exactly the set of categories S.}\end{align*}
    \item Initialize $dp[\emptyset]=0$ and apply the recursion \begin{align}dp[S]=\max_{x\in S}\left(dp[S\backslash \{x\}]+\sum_{x'\in S\backslash \{x\}}H[x',x]\right)\label{eq:dp}\end{align} that allows to iteratively compute $dp[S]$ for sets $S$ of increasing cardinality.
    \item Denote the maximum number of concordant pairs by $C=dp[\{x_1, ..., x_k\}]$ and report the result $$\widehat{\gamma}^{*}=\frac{2C-T}{T}$$
\end{enumerate}
The crucial step in this algorithm is step 4. It relies on the characteristic of the problem that including a new value $x$ with the hitherto largest number adds the same number of concordant pairs irrespective of the order that the previously added values have between each other. As a consequence, we break the optimization problem into $2^k$ smaller pieces of which each has a complexity of $\mathcal{O}(k)$ due to the sum in (\ref{eq:dp}). Step 5 follows by $C+D=T$ for a fixed (at least ordinal) sample with $D$ denoting the number of discordant pairs and the ubiquitous definition $\widehat{\gamma}=(C-D)/(C+D)$. This algorithm is implemented in the corresponding \texttt{R} package \texttt{NCor} (\texttt{https://github.com/jan-lukas-wermuth/NCor}) and roughly allows for $k=30$ if only the estimator is desired and $k=6$ if on top of that the independence test is executed. 

Case 2 is much more complicated from a discrete optimization point of view. While it is definitely possible to construct heuristic approaches that approximate the maximum, the only algorithm that I am aware of which always finds the true maximum is indeed brute force. However, calculating the number of concordant pairs for a fixed numbering can be done efficiently. More precisely, we denote $n_{ij}$ with $i\in \{1, ..., k\}$ and $j\in \{1, ..., l\}$ as the absolute frequency associated with the (bivariate) realization after numbering $(i,j)$. Equivalently, it can be understood as the $(i,j)$-entry in an ordered contingency table that has been created after assigning the permutations. One can then compute $B_{ij}:=\sum_{k>i}n_{kj}$ for all $i\in \{1, ..., k\}$ by using the recursion $B_{ij}=B_{i+1,j}+n_{i+1,j}$ for all $j\in \{1, ..., l\}$. The same logic applies to $N_{ij}^{SE}:=\sum_{l>j}B_{il}$, which can be calculated by $N_{ij}^{SE}=N_{i,j+1}^{SE}+B_{i,j+1}$. $N_{ij}^{SE}$ amounts to the total number of observations for which both coordinates are strictly larger than $i$ or $j$, respectively. Heuristically, it is the sum of all the entries in the ordered contingency table that are south-east (hence the superscript $SE$) to the $(ij)$-entry. The total number of concordant pairs is then determined by $$C=\sum_{i=1}^k\sum_{j=1}^ln_{ij}N_{ij}^{SE}.$$
Note that this entire algorithm has complexity $\mathcal{O}(kl)$ and is therefore independent of the sample size. Yet, it has to be performed for every permutation and has therefore complexity $\mathcal{O}(k\cdot l\cdot k!\cdot l!)$. A similar algorithm can be applied for the discordant pairs and the result follows. In practice, the \texttt{R} implementation \texttt{NCor} yields fast results for the estimator roughly until $k=l=8$. If the independence test is also desired, results are quickly available until $k=l=4$.
\clearpage
\section{Additional Graphs}\label{sec-App:addGraphsTables}
\begin{figure}[h]
    \centering
    \includegraphics[width=\linewidth]{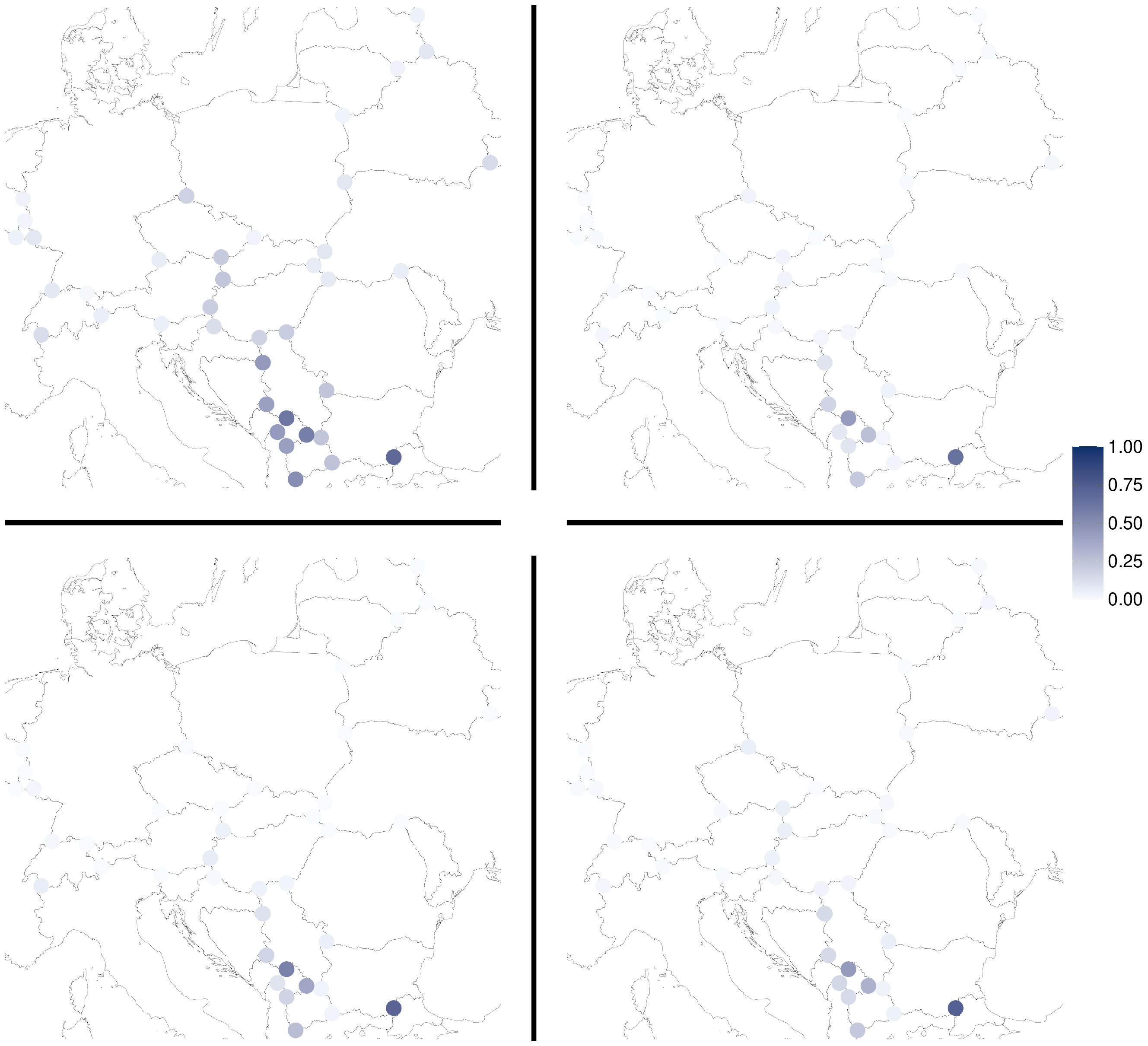}
    \caption{This figure shows $PC$ (upper left plot), $\tau$ (upper right plot), $\lambda$ (lower left plot) and $U$ (lower right plot) for the dependence between the variables country and religion. Each colored dot represents the value of the respective coefficient if the variable country is restricted to the three countries featuring in the border triangle at which the dot is located. The variable religion is always restricted to the three monotheistic world religions Christianity, Islam and Judaism.}
    \label{fig:Rel_PC_tau_lambda_uncert}
\end{figure}

\end{document}